\documentclass[runningheads]{llncs}
\usepackage[T1]{fontenc}
\usepackage{graphicx}
\usepackage{amsfonts,amsmath,amssymb}
\usepackage{booktabs}
\usepackage{algorithm}
\usepackage{algorithmic}
\usepackage{todonotes}

\usepackage{tikz}
\usepackage{tcolorbox}

\def\NP{\textsf{NP}}
\def\P{\textsf{P}}

\def\XP{\textsf{XP}}
\def\cost{\textsf{cost}}
\def\OPT{\textsf{OPT}}
\def\I{\mathcal{I}}
\def\fixedenvymin{\textsc{Temporal-EnvyMin}}
\def\variableenvymin{\textsc{Temporal-EnvyMin-Seq}}

\DeclareMathOperator*{\argmin}{arg\,min}

\tcbset{size=small,top=1mm,bottom=1mm,middle=1mm}

\begin{document}
\title{Minimizing Cumulative Envy in Allocating a Sequence of Items\thanks{A preliminary version of this paper was presented at the 19th International Symposium on Algorithmic Game Theory (SAGT), 2026.}}
%
%
\author{Paul W. Goldberg \and
Isaac Robinson \and
Nicholas Teh}
%
%
\institute{University of Oxford, UK \\
\email{\{paul.goldberg,isaac.robinson,nicholas.teh\}@cs.ox.ac.uk}\\
}

\maketitle              
\begin{abstract}
    We study temporal fair division with indivisible goods that arrive sequentially and must be allocated irrevocably. In contrast to the usual online model, we assume that valuations and future arrivals are known in advance, and ask how unfairness evolves during the process. We introduce \emph{cumulative maximum envy}: the sum, over all rounds, of the maximum pairwise envy at that round. Equivalently, this is the area under the worst-envy curve, and it captures both the magnitude and the duration of envy. For a fixed arrival order, we show that the corresponding decision problem is strongly NP-complete and that minimizing this objective admits no constant-factor approximation unless P = NP, even under identical valuations and even under binary valuations. We complement these hardness results with a dynamic program that gives pseudopolynomial-time solvability for a constant number of agents, polynomial-time algorithms in further restricted settings, and an FPTAS for fixed $n$ under identical integer valuations. We then study a sequencing variant where the algorithm may choose the arrival order. This variant remains NP-complete even for two agents with identical valuations; however, a simple greedy algorithm achieves a $3/2$-approximation for $n=2$ agents, an $n/(n-1)$-approximation for any number of agents, and an additive guarantee depending on the maximum value of any good.
\end{abstract}

\section{Introduction}
Fair division of indivisible goods is a fundamental, well-studied problem at the intersection of computer science and economics \cite{brams1996fairdivision,moulin2003fairdivision}, with applications ranging from divorce settlements to distributed load balancing.
A central fairness notion in this context is \emph{envy-freeness}, which requires that every agent values their allocated bundle at least as much as any other agent's bundle. 
However, this strong requirement is not always achievable; consider, for example, the case of two agents and a single good that both value positively, which must be allocated to one of them.
As a result, recent years have seen considerable interest in relaxations of envy-freeness. 
Arguably, the two most prominent relaxations are \emph{envy-freeness up to one good (EF1)}~\cite{budish2011combinatorial,lipton2004approximately} and \emph{envy-freeness up to any good (EFX)}~\cite{caragiannis2019unreasonable,gourves2014near}.
An allocation is EF1 if, for any pair of agents, any envy can be eliminated
by removing one good from the envied agent's bundle, where goods are compared
using the envious agent's valuation; in particular, removing a most-valued good
is sufficient. It is EFX if the envy is eliminated even after removing a
least-valued good.

The concept of EF1 originates from the seminal work of Lipton et al.~\cite{lipton2004approximately}, which introduced the first algorithmic approach for computing approximately envy-free allocations of indivisible goods. Their \emph{envy-cycle elimination} algorithm remains a foundational tool and has since been used as a building block in tackling some of the most challenging problems in fair division, including the recent breakthrough on the existence of EFX allocations for three agents~\cite{chaudhury2024efxthreeagents}. Interestingly, EF1 emerged from an effort to bound envy: if each agent's maximum envy can be bounded by the largest value of a single good, then a polynomial-time algorithm can achieve this relaxed form of envy-freeness. This same line of reasoning also motivates the optimization problem of minimizing envy, which was considered by Lipton et al.~\cite{lipton2004approximately} (see also~\cite{nguyen2014minimizing}). The optimization perspective has received comparatively less attention than discrete relaxations such as EF1 and EFX.

While most prior works study the fair division problem in the \emph{offline} setting (i.e., all goods are immediately available and are allocated to agents in a single-shot), 
many real-world systems do not allocate all resources at once. Instead, indivisible goods become available over time and must be assigned as they appear: affordable housing units are completed in phases, public extraction permits are issued across multiple periods, and platforms assign arriving tasks or opportunities to workers. In such settings, stakeholders care not only about the \emph{final} allocation, but also about the \emph{trajectory}: large early imbalances can create lasting disadvantages, distort downstream behavior, or erode trust.

Several recent works have considered an \emph{online} model where goods arrive one at a time, and must immediately and irrevocably be allocated to agents.
This captures numerous application scenarios where goods need to be allocated in a sequential manner, such as assigning donations to food banks \cite{Aleksandrov2015}, allocating resources in cloud computing environments \cite{bei2022cloudcomputingallocation}, or allocating impressions in recommendation systems \cite{murray2020robustuncertain}.
In the standard online fair division model, it is usually assumed that each good's
valuation is not known until its arrival. 
Most work on online fair division evaluates fairness only at the end of the process, once all goods have arrived and been allocated. In practice, agents experience the allocation process as it unfolds, and limiting unfairness throughout the process can be important for maintaining balance, participation, and trust, especially in long-term allocation systems.

This motivation is central to recently-studied \emph{temporal fair division} model \cite{cookson2025temporalfairdivision,elkind2025temporal}, where the goal is to guarantee fairness not merely at the end, but at \emph{every prefix} during the allocation process given information about future arrivals and valuations for each good-agent pair.
This perspective provides a different lens on sequential fair division.
In the standard and commonly-studied online model, agents have limited to no information about future arrivals, and fairness is evaluated only at the end (after all goods have arrived).
In contrast to these standard online models with unknown future arrivals, the temporal fair division model assumes complete information about valuations and future arrivals while striving for typically stronger, \emph{cumulative} fairness guarantees \cite{cookson2025temporalfairdivision,elkind2025temporal}.\footnote{Note that without this ``cumulative fairness'' consideration, the problem reduces to the standard offline model.} Consider the allocation of pre-built affordable housing units that become available over time: while the housing authority knows when each unit will be ready, ensuring equitable distribution throughout the multi-year process—not just at the end—is essential for maintaining public trust and preventing systematic disadvantage to any group. Similarly, when allocating extraction rights for public natural resources across multiple time periods to private companies, temporal fairness ensures that no firm gains disproportionate early advantages that could distort market competition.

To quantify experienced unfairness over time, we take an optimization view rooted in classic envy minimization~\cite{lipton2004approximately}. For an allocation process, let $A_i^t$ denote agent $i$'s bundle after the
first $t$ goods, and define $\mathrm{envy}_{i,j}^t
:= \max\{0,\,v_i(A_j^t)-v_i(A_i^t)\}$.
Let $c_t:=\max_{i,j\in N}\mathrm{envy}_{i,j}^t$ be the maximum envy after the first $t$ goods have been allocated. The sequence $(c_1,\dots,c_m)$ is the worst-envy curve of the process. Our objective is $\sum_{t=1}^m c_t$, the area under this curve.\footnote{In online fair division, the term \emph{cumulative envy} is sometimes used for the maximum envy after the first $t$ goods have been allocated (see, e.g.,~\cite{benade2018envyvanish}), which is $c_t$ in our notation. Our objective instead adds up the maximum envy over all rounds.} 
This objective is a simple, parameter-free baseline: it uses the standard maximum-envy measure at each round and aggregates it over time, so it penalizes both large envy and envy that persists for many rounds. We do not claim that this is the only meaningful temporal objective. Rather, it is a natural optimization counterpart to prefix-based temporal fairness notions, and it remains meaningful even when maintaining a strong fairness guarantee at every prefix is impossible.

For each agent $i$, let $e_i^t := \max_{j\in N}\mathrm{envy}_{i,j}^t$ and $E_i := \sum_{t=1}^m e_i^t$.
Since $c_t=\max_{i\in N}e_i^t$, we have $\max_{i\in N} E_i
\le \sum_{t=1}^m c_t
\le \sum_{i\in N} E_i
\le n\max_{i\in N} E_i$.
Thus, our objective upper-bounds the cumulative envy experienced by every
agent and is within a factor $n$ of the largest individual cumulative envy.
In either of our two models, an optimal solution for our objective is therefore
an $n$-approximation for directly minimizing $\max_{i \in N} E_i$ over the same
feasible allocation processes. However, our objective need not distribute
envy evenly across agents. Directly minimizing $\max_{i \in N} E_i$, $\sum_{i \in N} E_i$,
or weighted versions of these quantities is a natural direction for future
work.
It would also be useful to study objectives that directly minimize the largest or total individual cumulative envy.

We investigate cumulative maximum envy minimization in two related settings. First, we study the standard temporal model where the arrival order is fixed and may be adversarially chosen as part of the input. Second, we study a sequencing variant in which the algorithm may choose the
order of the goods before allocating them. This models settings where all goods
are known and can be held briefly before release; for example, a food bank may
have several donated pallets in storage and choose which pallet to allocate
next, or a public agency may choose the order in which approved permits are
issued.



\subsection{Our Contributions}
We begin in Section~\ref{sec:intractability} by presenting several intractability results. We show that the fixed-order problem is NP-complete even under identical valuations and even under binary valuations. In both cases, we also prove that no constant-factor approximation is possible unless $\P = \NP$.

Motivated by these hardness results, in Section~\ref{sec:tractability} we give a dynamic programming algorithm that computes an optimal allocation sequence in time $\mathcal{O}((m+1)(V_{\max}m^2+1)(mV_{\max}+1)^{n^2}\cdot n^3 )$,
where $n$ is the number of agents, $m$ is the number of goods, and $V_{\max}$ is the maximum value of any good to any agent. This gives pseudopolynomial-time solvability when $n$ is constant, and polynomial-time solvability under binary valuations for constant $n$. If the number of goods is constant, an optimal solution can also be found by direct enumeration. We further show that binary identical valuations admit a simple polynomial-time optimal algorithm, and that identical integer valuations admit an FPTAS for every fixed $n$. This contrasts with the offline setting, where the minimum maximum envy cannot be approximated within any factor unless $\P = \NP$, even for two agents with identical valuations~\cite{lipton2004approximately}; the difference is that, under identical valuations, our objective is always at least a $(1-1/n)$ fraction of the total value of the goods.

In Section~\ref{sec:variable}, we explore a natural variant of the model in which the algorithm is allowed to choose the arrival order of goods (while an adversary pre-selects the agents' valuations over the goods).
We show that this variant remains \NP-complete even for just two agents and identical valuations.
On the positive side, when agents have identical valuations, we design a polynomial-time algorithm that achieves:
(i) a $3/2$-approximation for $n=2$ agents,
(ii) a $n/(n-1)$-approximation for general $n$, and
(iii) a $\frac{n-1}{n} \, V_\text{max}$-additive approximation to the optimal cost.

\subsection{Related Work}
We highlight the works most relevant to our focus on envy minimization in the temporal setting.

\medskip \noindent
\textbf{Offline Fair Division.} 
In the offline fair division model, all goods are available simultaneously and allocated at once.
Lipton et al.~\cite{lipton2004approximately} investigate envy minimization in this setting, aiming to bound the maximum envy between any pair of agents.
They introduce the well-known envy cycle elimination algorithm, which runs in polynomial time and guarantees an allocation where the \emph{maximum envy} is at most the highest value of any good.
This result also gives rise to the well-known EF1 concept.
They also study the problem of computing an allocation with minimum maximum envy. Since deciding whether this minimum is zero is NP-hard even for two agents with identical valuations, no approximation within any factor is possible unless $\P = \NP$. They therefore also consider the \emph{envy-ratio}; under identical valuations, they relate it to scheduling on identical machines and give a PTAS, as well as an FPTAS when the number of agents is constant. Nguyen and Rothe~\cite{nguyen2014minimizing} show that minimizing the envy-ratio is APX-hard in general, and give an FPTAS for a constant number of agents with additive valuations.

\medskip \noindent
\textbf{Online Fair Division.} 
In the online fair division model, goods arrive sequentially in a fixed order, and the algorithm has no knowledge of future items.
Due to this uncertainty, most works on this model focus on ensuring fairness guarantees at the final round, once all goods have arrived and been allocated.
Several works look at minimizing the maximum envy in this setting.
Benad\`{e} et al.~\cite{Benade2023Fair} (see also~\cite{benade2018envyvanish}) study a fully adversarial setting in which item values are chosen adaptively: an adversary observes the algorithm and all decisions made thus far, and selects the value of each incoming good to maximize envy at the end.
They prove that, even with just two agents, any online algorithm must incur a maximum envy of $\Omega(m^{r/2})$ for every constant $r<1$, where $m$ is the number of goods, and that allocating each good to a uniformly random agent nearly matches this bound.

Recognizing that adaptive adversaries may be unrealistically powerful, Halpern et al.~\cite{HalpernPsomasVermaXie2025} consider an \emph{oblivious adversary} model, where item values are fixed in advance and chosen to be worst-case with respect to the algorithm, but do not depend on the actual allocation decisions made in previous rounds.
In this setting, they show connections to the \emph{online multicolor discrepancy} problem and, in the special case of two agents, to \emph{online vector balancing} \cite{spencer1977}.
Other works give the algorithm partial information about future goods~\cite{neoh2025online}, let it see a few goods ahead~\cite{amanatidis2025onlinefairdivisionpersonalized}, or let it hold and reorder a bounded number of goods before allocating them~\cite{amanatidis2026buffers}; the last two also give guarantees that hold after every allocated good.
For approximate proportionality, Choo et al.~\cite{choo2026approximate} show that several natural greedy rules cannot guarantee any positive approximation to proportionality up to one good (PROP1) against an adaptive adversary, and Neoh and Teh~\cite{neoh2026closing} show that this holds for every online algorithm, and also for proportionality up to $k$ goods with any fixed $k$; positive guarantees are possible against an oblivious adversary~\cite{choo2026approximate,cohen2026oblivious} or when each agent's largest value for a good is known in advance~\cite{choo2026approximate,neoh2026closing}. Online fair division under budget constraints has also been studied~\cite{cohen2026budget}.

In contrast, our work assumes full information about the values of all goods from the outset, and focuses on cumulative envy minimization, aligning more with the temporal fair division model than with adversarial online models.

\medskip \noindent
\textbf{Temporal Fair Division.}
In the temporal fair division model (sometimes also referred to as the \emph{fully informed} online fair division model), the sequence of goods is known in advance, allowing the algorithm to make allocation decisions with full knowledge of future arrivals.
This setting enables us to introduce the study of fairness guarantees that are cumulative over time, rather than only at the final round.
Cookson et al.~\cite{cookson2025temporalfairdivision} introduce temporal fair division and give existence/algorithmic guarantees for achieving EF1-type fairness simultaneously at multiple points in time, with stronger positive results in special cases (e.g., two agents or structured preferences).
Elkind et al.~\cite{elkind2025temporal} formalize temporal EF1 (TEF1) for sequential item arrivals, show that deciding whether a TEF1 allocation exists is NP-hard, and provide polynomial-time algorithms in several structured cases.
He et al.~\cite{he2019achieving} maintain EF1 at every round by allowing limited reallocations (or what they call adjustments), and quantify how many such changes are necessary (including separations between informed vs. uninformed settings); in particular, with full information about future goods, no adjustments are needed for two agents, whereas for three or more agents EF1 cannot always be maintained at every round without adjustments, so TEF1 allocations need not exist.
Choi and Li~\cite{choi2026structured} study temporal versions of EFX and maximin share fairness, and Choi and Li~\cite{choi2026scheduling} further allow each good to be delayed by a bounded number of rounds before it is allocated. Choi et al.~\cite{choi2026mixed} identify settings with goods and chores in which TEF1 allocations exist and can be computed efficiently.
In contrast to this constraint-based line of work (maintaining temporal EF1 at every prefix), we take an optimization view and minimize the cumulative ``area under the envy curve'', which remains meaningful even when strong prefix-fairness constraints are infeasible.
Fairness over time has also been studied in repeated allocation, where the same set of items is allocated in every round~\cite{igarashi2023repeatedfairallocation,caragiannis2024repeatedly}.
Micheel and Wilczynski~\cite{micheel2024fairness} study a related repeated house allocation model in which every round is a complete matching, so each agent receives one object. They focus on envy-based fairness across the repeated assignments, together with Pareto efficiency in each round. Their model differs from ours in two main ways: they make a complete assignment in every round, whereas we allocate one indivisible good irrevocably at each step; and they study fairness conditions across repeated assignments, whereas we minimize cumulative maximum envy.
Lim et al.~\cite{lim2026repeated} study repeated matchings that maximize the minimum cumulative utility of an agent, either at the end or at every round.
Fairness over time has also been studied in other collective decision problems, such as temporal slot assignment~\cite{elkind2022slot} and perpetual or temporal voting~\cite{lackner2020perpetual,elkind2024temporalvoting}.

\section{Preliminaries} \label{sec:prelims}

For each positive integer $z$, we denote $[z] := \{1,\dots,z\}$.
An instance of the \emph{fair division} problem is a tuple $\I = (N, G, \mathbf{v} = (v_1,\dots,v_n))$, where $N = [n]$ is the set of (fixed) \emph{agents}, $G$ is a set of $m$ goods, and for each $i \in N$, the function $v_i\colon 2^G \rightarrow \mathbb{R}_{\geq 0}$ is the \emph{valuation function} of agent $i$.
We write $v$ instead of $v_i$ when all agents have identical valuation functions.
In the \emph{temporal fair division} problem, goods arrive sequentially according to some order $\pi$, and we must allocate each good (to an agent) as it arrives. Let $\Pi$ denote the space of all such arrival orders.
In this work, we consider both the setting where $\pi$ is chosen by an adversary (i.e., the standard temporal fair division setting) and the setting where $\pi$ can be chosen by the algorithm.

As with most works in the fair division literature, we assume that valuation functions are \emph{additive}, i.e., for any subset of goods $B \subseteq G$, $v_i(B) = \sum_{g \in B} v_i(\{g\})$.
For convenience, we write $v_i(g)$ instead of $v_i(\{g\})$ for a single good~$g$.
Every subset of goods in $G$ can also be referred to as a \emph{bundle}.
An \emph{allocation} $\mathcal{A} = (A_1,\dots, A_n)$ is a partition of the goods into bundles, with agent $i \in N$ receiving bundle $A_i$.
For each $i \in N$ and $t \in [m]$, let $A_i^t$ denote the partial allocation of agent $i$ \textit{after} the $t$-th arriving good is allocated to an agent and let $A_i^0 = \varnothing$ be the initial empty allocation, and $A_i^m = A_i$.
Equivalently, we can represent an allocation by an allocation sequence $S = (s_1,\dots, s_m) \in N^{m}$, where for each $t \in [m]$, the $t$-th entry $s_t$ specifies the agent that the $t$-th arriving good is allocated to.
Let the space of all such sequences be $\mathcal{S}$.

Next, we define \emph{envy} and the \emph{envy minimization} objective.
Let $S$ be the allocation sequence corresponding to $\mathcal{A}$, given some order $\pi$.
Then, for any $t\in[m]$, the \emph{envy} of agent $i\in N$ toward agent $j\in N\setminus\{i\}$ is $\mathrm{envy}_{i,j}^t(S,\pi)
:= \max\{0,\,v_i(A_j^t)-v_i(A_i^t)\}$;
in particular, $\mathrm{envy}_{i,i}^t(S,\pi)=0$.
The \emph{per-round cost} is $c_t(S,\pi)
:= \max_{i,j\in N}\mathrm{envy}_{i,j}^t(S,\pi)$,
and the \emph{cumulative maximum envy} is $\cost(S,\pi)
:= \sum_{t=1}^m c_t(S,\pi)$,
which we seek to minimize.

Maximum envy is a standard quantitative objective in both offline envy minimization~\cite{lipton2004approximately} and online envy minimization~\cite{benade2018envyvanish,HalpernPsomasVermaXie2025}. In most online models, fairness is evaluated only after the last item has been allocated. In the fully informed temporal model, however, future arrivals and valuations are known, so it is natural to ask for fairness over the trajectory of the allocation, not just at the end.
One stronger requirement would be to optimize every prefix simultaneously (as in the anytime optimality notion of Lim et al.~\cite{lim2026repeated} for repeated matching), or to impose a fairness constraint at every prefix as in temporal EF1~\cite{elkind2025temporal}. We instead optimize the total area under the worst-envy curve. This permits short-term imbalances when they are compensated later, but penalizes imbalances that are large or long-lasting. This matches settings where temporary disparities are acceptable but prolonged severe unfairness is not.

Finally, optimizing only the final maximum envy can be badly misaligned with minimizing cumulative exposure to envy. Example~\ref{eg_finalvscumulative} gives a linear separation even under identical valuations; conversely, cumulative maximum envy always controls final envy because $c_m(S,\pi)$ is one of the terms in $\sum_t c_t(S,\pi)$.

\begin{example} \label{eg_finalvscumulative}
    Consider an instance with $n=2$ and identical valuation function $v$.
    Fix any integer $q\ge 2$, let $m:=2q+1$ and $\varepsilon:=\frac{99}{2q}$.
    Then $m$ is odd and $(m-1)\varepsilon=99$. Let
    $\pi=(g_1,\dots,g_m)$ be a fixed arrival order where
    $v(g_t)=\varepsilon$ for every $t\le m-1$ and $v(g_m)=100$.
    For $n=2$, we have $c_t(S,\pi)=\left|v(A_1^t)-v(A_2^t)\right|$ for every $t \in [m]$.
    
    \begin{itemize}
    \item \emph{An allocation minimizing the final round cost.}
    Allocate $\{g_1,\dots,g_{m-1}\}$ to agent $1$ and $g_m$ to agent $2$.
    Then $v(A_1^m)=99$ and $v(A_2^m)=100$, so $c_m(S,\pi)=1$
    (and no allocation can have $c_m(S,\pi)<1$ since $v(G)=199$).
    Moreover, for each $t\le m-1$ we have $v(A_1^t)=t\varepsilon$ and $v(A_2^t)=0$,
    so $c_t(S,\pi)=t\varepsilon$. Thus, $\mathrm{cost}(S,\pi)=\sum_{t=1}^{m-1} t\varepsilon + 1
    =\varepsilon\frac{m(m-1)}{2}+1
    =\frac{99m}{2}+1$, which grows linearly in $m$.
    
    \item \emph{An allocation with much smaller cumulative cost.}
    Alternately allocate the first $m-1$ goods between the agents.
    Then for every $t\le m-1$, we maintain $| v(A_1^t)-v(A_2^t)|\le \varepsilon$,
    so $c_t(S,\pi)\le \varepsilon$; more precisely, $c_t(S,\pi)=\varepsilon$ on odd $t$
    and $c_t(S,\pi)=0$ on even $t$.
    Since $m-1$ is even, $\sum_{t=1}^{m-1} c_t(S,\pi)=\frac{m-1}{2}\varepsilon$.
    Allocating $g_m$ gives $c_m(S,\pi)=100$, so $\mathrm{cost}(S,\pi)=\frac{m-1}{2}\varepsilon+100
    =\frac{99}{2}+100=\frac{299}{2}
    <\frac{99m}{2}+1$ for all $m > 3$.
    \end{itemize}
\end{example}

\section{Intractability Results}\label{sec:intractability}
We first present several results highlighting the computational intractability of envy minimization as an objective in the temporal fair division setting. 

We begin by formalizing the decision problem.
\begin{tcolorbox}[title=\fixedenvymin]
    Given an instance $\I = (N, G, \mathbf{v})$, an order $\pi = (g_1,\dots,g_m)$ over the goods in $G$, and an integer $\gamma \in \mathbb{Z}_{\geq 0}$, can we find an allocation sequence $S$ such that $\textsf{cost}(S,\pi) \leq \gamma$?
\end{tcolorbox}
Then, we show that the problem is computationally intractable even in the restricted setting where agents have \emph{identical valuation functions}---that is, for any pair of agents $i,j \in N$ and good $g$, we have $v_i(g) = v_j(g) = v(g)$.
Moreover, we go one step further by showing that a constant-factor approximation algorithm is unlikely to exist.
\begin{theorem} \label{thm:intractable_identical}
    \fixedenvymin{} is strongly \NP-complete even when agents have identical valuations.
    Moreover, in this setting, there exists no constant-factor approximation algorithm, unless $\P = \NP$.
\end{theorem}
\begin{proof}
We reduce from the \textsc{3-Partition} problem, which is known to be strongly NP-complete \cite{garey1975complexitymultiprocessor}.
    An instance of \textsc{3-Partition} is given by a multiset of positive integers $X = \{x_1,\dots,x_{3k}\}$ (with $\sum_{j = 1}^{3k} x_j = kT$); it is a yes-instance if there exists a partition of $X$ into $k$ subsets $X'_1,\dots,X'_k$ (each containing three integers) such that $\sum_{x \in X'_j} x = T$ for all $j \in [k]$, and a no-instance otherwise.
    Note that for each integer $x \in X$, we can without loss of generality assume that $x < T$ (otherwise it is trivially a no-instance).

    Construct an instance of our problem, where $N = [k + 2]$ and we have $m := 9k+7$ goods arriving in the order $\pi = (g_1,g_2,\dots, g_{3k+1}, g_{3k+2}, g_{3k+3}, \dots, g_{9k+7})$, with valuations defined as follows:
    
    \begin{equation*}
    v(g_j) = 
        \begin{cases} 
      10T & \text{ if } j = 1; \\
      x_{j-1} + 3T & \text{ if }  j \in \{2,\dots,3k+1\}; \\
      10T & \text{ if }  j = 3k+2; \\
      0 & \text{ if } j \in \{3k+3,\dots,9k+7\}.
      \end{cases}
    \end{equation*}

    We first show that \fixedenvymin{} is NP-complete (it is easy to see that the problem is in NP, we thus focus on proving NP-hardness).
    We do so by proving that there exists a solution $X'_1,\dots, X'_k$ to \textsc{3-Partition} if and only if there exists an allocation $\mathcal{A}$ (and hence corresponding allocation sequence $S$) such that $\textsf{cost}(S,\pi) \leq 10T \cdot (3k+1)$ for the above constructed instance.

    ($\Rightarrow$) Suppose there exists a solution $X'_1,\dots, X'_k$ to \textsc{3-Partition}.
    Consider the sequence $S = (s_1,\dots,s_m)$: 
    \begin{equation*}
        s_j = 
        \begin{cases}
            1 & \text{ if } j \in \{1, 3k+3, \dots, 9k+7\}; \\
            2 & \text{ if } j = 3k+2; \\
            \ell + 2 & \text{ if } x_{j-1} \in X'_\ell \text{ for } \ell \in [k]
        \end{cases}
    \end{equation*}
    Observe that for all $t \in \{1,\dots, 3k+1\}$, every agent $i \geq 3$ holds at most three positively-valued goods, whose total value is at most $T + 9T = 10T$. Thus, $\max_{i \in N} v(A_i^t) = 10T$ and $\min_{i \in N} v(A_i^t) = 0$, so we have that $c_t(S,\pi) = 10T$. 
    Moreover, since each agent in $N \setminus \{1,2\}$ has bundle value $T + 3 \cdot 3T = 10T$, we get that $c_{3k+2}(S,\pi) = 0$.
    Then,
    \begin{align*}
        \textsf{cost}(S,\pi) & = \sum_{t=1}^{3k+2} c_t(S,\pi) + \sum_{t=3k+3}^{9k+7} c_t(S,\pi)\\
        & = 10T \cdot (3k+1) + 0 + 0= 10T \cdot (3k+1).
    \end{align*}
    
    ($\Leftarrow$) Suppose there does not exist a solution to \textsc{3-Partition}.
    This means that there does not exist a partition of integers in $X = \{x_1,\dots, x_{3k}\}$ into $k$ subsets of equal sum $T$.
    Equivalently, there is no way to partition the $3k$ goods $\{g_2,\dots,g_{3k+1}\}$ into $k$ bundles each of total value $10T$.
    
    Let $i^* \in N$ be the agent that receives good $g_{3k+2}$, i.e., $g_{3k+2} \in A_{i^*}$.
    We split our analysis into two cases.
    
    \paragraph{Case 1: No good in $\{g_1,\dots,g_{3k+1}\}$ allocated to agent~$i^*$.}
        Let $g_1$ be allocated to agent $i' \in N \setminus \{i^*\}$. Then, $A_{i'}^1 = \{g_1\}$ and $v(A_{i'}^1) = 10T$.
        Now, since no goods in $\{g_1,\dots,g_{3k+1}\}$ was allocated to agent~$i^*$, for all $t \in \{1,\dots, 3k+1\}$, we have that $A_{i^*}^t = \varnothing$, giving us $c_t(S,\pi) \geq 10T$.
        At round $3k+2$, agent $i^*$ receives a good valued at $10T$ and thus has value exactly $10T$ (since she received nothing earlier).
        If $c_{3k+2}(S,\pi)=0$, then every agent must have value exactly $10T$ at round $3k+2$ (since the total value of the first $3k+2$ goods is $10T(k+2)$, equality across $k+2$ agents forces each to have value $10T$).
        The agent $i'$ holding the first good valued at $10T$ (i.e., $g_1$) cannot hold any other positively-valued good (each such good has value $>3T$), else it would exceed $10T$.
        Therefore, the remaining $k$ agents must partition the $3k$ goods of values $x_{j-1}+3T$ into $k$ bundles each of value $10T$. 
        Since each such good has value strictly between $3T$ and $4T$, each bundle must contain exactly $3$ goods, and the $x_j$ parts in each bundle must sum to $T$, which is a \textsc{3-Partition} solution, giving us a contradiction.
        Thus, $c_{3k+2}(S,\pi)\ge 1$, and together with $\sum_{t=1}^{3k+1} c_t(S,\pi) \ge 10T(3k+1)$, we get $\textsf{cost}(S,\pi)>10T(3k+1)$.
        
    \paragraph{Case 2: Some good in $\{g_1,\dots,g_{3k+1}\}$ allocated to agent~$i^*$.} 
        Let $g_1$ be allocated to agent $i' \in N$ (note that it could be $i' = i^*$).
        Since $i^*$ is also allocated good $g_{3k+2}$, we have that $v(A_{i^*}^{3k+2}) \geq 3T + 10T = 13T$.
        If $i'=i^*$, then $3k$ goods are distributed among $k+1$ agents, so some agent gets at most two goods from $\{g_2,\dots, g_{3k+1}\}$.
        If $i' \neq i^*$, then $i^*$ must have taken at least one of the $3k$ goods (since $g_1$ went to $i'$), leaving at most $3k-1$ such goods for the remaining $k$ agents in $N\setminus\{i',i^*\}$; and by pigeonhole one of them gets at most two goods from $\{g_2,\dots, g_{3k+1}\}$.
        Now, for each $j \in \{2,\dots,3k+1\}$, we know that $v(g_j) < T + 3T = 4T$ (by our assumption that $x < T$ for each $x \in X$).
        Then, $v(A_i^{3k+2}) < 8T$.
        This means that $c_{3k+2}(S,\pi) \geq v(A_{i^*}^{3k+2}) - v(A_i^{3k+2}) > 13T - 8T = 5T$.
        There are $6k+5$ zero-valued goods after round $3k+2$, i.e., rounds $t = 3k+3,\dots, 9k+7$.
        Since these goods do not change any bundle values, $c_t(S,\pi)= c_{3k+2}(S,\pi)$ for all such $t$.
        Thus,
        \begin{align*}
            \sum_{t'=3k+3}^{9k+7} c_{t'} & = (6k+5) \cdot c_{3k+2}(S,\pi) > (6k + 5) \cdot 5T > 10T(3k+1),
        \end{align*}
        which implies $\text{cost}(S,\pi) > 10T(3k+1)$.

    Next, we argue that in (the same) setting where agents have identical valuations, there exists no constant-factor approximation algorithm, unless P = NP.
    
    Since \textsc{3-Partition} is strongly NP-complete, meaning even when $T = \textsf{poly}(n)$, the problem remains NP-complete. 
    Thus, we can assume $T = \textsf{poly}(n)$.

    Then, note that in the reduction above, we showed that given an instance $\I$ and order $\pi$, determining whether we can find an allocation sequence $S$ such that $\textsf{cost}(S,\pi) \leq 10T \cdot (3k+1)$ is NP-complete.
    Equivalently,
    \begin{itemize}
        \item any yes-instance admits a sequence $S$ with $\cost(S,\pi) \leq 10T \cdot (3k+1)$, and in particular $c_{3k+2}(S,\pi) = 0$;
        \item any no-instance will force any sequence $S$ to have $c_{3k+2}(S,\pi) \geq 1$.
    \end{itemize}
    Let $B:=10T(3k+1)$ and fix any constant $\alpha\ge 1$. To amplify the gap, append $Z_\alpha:=\left\lceil B(\alpha-1)\right\rceil+1$ zero-valued goods after $g_{9k+7}$. In a yes-instance, the allocation above still has cost at most $B$, because the envy is already zero from round $3k+2$ onward. In a no-instance, every allocation has positive envy after round $3k+2$; since all appended goods have value zero, this envy persists through the appended rounds. 
    Thus, every allocation has cost strictly larger than $B+Z_\alpha > \alpha B$.
    Therefore, given a \textsc{3-Partition} instance, run the assumed
    $\alpha$-approximation on the corresponding \textsc{Temporal-EnvyMin}
    instance and compare its returned cost with $\alpha B$. If the original
    \textsc{3-Partition} instance is a yes-instance, the returned cost is at most
    $\alpha B$. If it is a no-instance, every feasible allocation has cost greater
    than $\alpha B$. The assumed approximation algorithm would therefore decide
    \textsc{3-Partition} in polynomial time, which is impossible unless
    $\P = \NP$.
\end{proof}

Notably, the above hardness result continues to hold even when the maximum value of any good is bounded by a polynomial in $n$, owing to the strong NP-completeness property of the \textsc{3-Partition} problem used in the reduction.

In the offline setting, the minimum maximum envy cannot be approximated within any factor simply because it may be zero~\cite{lipton2004approximately}. This argument does not apply directly to our objective under identical valuations, since the first good with positive value always creates envy, so even yes-instances have positive optimal cost. Our reductions instead end with zero-valued goods, so that any envy remaining after the last positively valued good is counted in each of the remaining rounds.

We now turn to a different restricted setting in which agents have \emph{binary valuation functions}---that is, for each agent $i \in N$ and good $g$, $v_i(g) \in \{0,1\}$.
In this case, we establish a second intractability result by reducing from a known NP-hard problem in the standard (offline) fair division literature: deciding whether there exists a complete envy-free allocation under binary valuations, a problem we refer to as \textsc{Binary-EF}.
This hardness was first shown implicitly by Aziz et al.~\cite{aziz2015fairassignmentordinal} and later made explicit by Hosseini et al.~\cite{hosseini2020fairdivisionwithholding}.
As in the previous case, we show that a constant-factor approximation algorithm is unlikely to exist.
\begin{theorem} \label{thm:intractable_binary}
    \fixedenvymin{} is \NP-complete even when agents have binary valuations.
    Moreover, in this setting, there exists no constant-factor approximation algorithm, unless $\P = \NP$.
\end{theorem}

\section{Tractability in Restricted Settings} \label{sec:tractability}
Given the computational hardness of finding an envy-minimizing allocation sequence, even under two distinct restricted settings (identical valuations or binary valuations), a natural follow-up question is whether there exists any restricted setting for which an optimal solution can be computed efficiently.

We begin with a simple warm-up result.
The most favorable remaining setting is when valuations are both binary \emph{and} identical.
While this does not exactly reduce the problem to one with identically-valued goods,\footnote{This setting has also been studied in other fair division contexts; see, e.g., Elkind et al.~\cite{Elkind2024choresbudget}.} it is closely related.
In this case, a straightforward greedy algorithm (which allocates each arriving good to an agent with the currently lowest total bundle value) produces an allocation with minimum cost.
This greedy rule is a special case of the algorithm that Elkind et al.~\cite{elkind2025temporal} use to obtain TEF1 allocations under generalized binary valuations; we show that, under binary and identical valuations, it also minimizes the maximum envy at every round.
\begin{proposition} \label{prop:tractability_binaryidentical}
    Under binary and identical valuations, an optimal sequence can be computed in polynomial time.
\end{proposition}

Next, for the main result of this section, we present a dynamic programming algorithm that runs in time $\mathcal{O}( (m+1)(V_{\max}m^2+1)(mV_{\max}+1)^{n^2}\cdot n^3 )$, where $V_\text{max} := \max_{i \in N, g \in G} v_i(g)$ is the maximum value of any good in any agent's view (and we can assume valuations are encoded in unary and scaled to integers).
For $t \in \{0,\dots,m\}$, let $M^t \in \{0,\dots, m\cdot V_\text{max}\}^{n\times n}$ be the \emph{valuation matrix} after the first $t$ goods have arrived: $M_{i,j}^t$ is the value agent~$i$ has for agent $j$'s bundle thus far, and $M^0$ is the all-zero matrix.

The algorithm is presented as follows (Algorithm~\ref{alg:DP}).
\begin{algorithm}[h]
\caption{Dynamic programming algorithm that returns an optimal sequence}
\begin{flushleft} 
    {\bf Input}: Instance $\I = (N, G, \mathbf{v})$ and order $\pi = (g_1, \dots, g_m)$. \\
    {\bf Output}: Allocation sequence $S = (s_1,\dots,s_m)$.
\end{flushleft}
\begin{algorithmic}[1]
\STATE Let $M^0$ be the all--zero $n\times n$ matrix
\STATE Create hash maps $\text{DP}[0],\dots,\text{DP}[m]$ mapping keys $(c,M)$ to predecessor pointers
\STATE Initialize $\mathrm{DP}[0][(0,M^0)] \gets \varnothing$

\FOR{$t = 0,\dots, m-1$}
  \FORALL{states $(c,M)$ in $\mathrm{DP}[t]$}
    \FOR{$s=1, \dots, n$}
      \STATE $M' \gets M$ (deep copy)
      \FOR{$i=1, \dots, n$}
        \STATE $M'_{i,s} \gets M'_{i,s} + v_i(g_{t+1})$
      \ENDFOR
      \STATE $c' \gets c + \max_{i,j} \max\{0, M'_{i,j}-M'_{i,i}\}$ \label{alg:line:cumulativecost}
      \IF{$(c',M') \notin \text{DP}[t+1]$} \label{alg:line:conditionfilter}
        \STATE $\mathrm{DP}[t+1][(c',M')] \gets ((c,M), s)$ 
      \ENDIF
    \ENDFOR
  \ENDFOR
\ENDFOR
    \STATE Let $(c^*,M^*):=\argmin_{(c, M) \in \text{DP}[m]} c$ \label{line:alg:DP_c*M*}
\STATE Initialize $S = (s_1,\dots,s_m) = (0, \dots, 0)$ and $(c,M)\gets (c^*,M^*)$
  \FOR{$t = m, \dots, 1$} \label{line:alg:witnessing_start}
    \STATE $((c_{\text{prev}},M_{\text{prev}}), s_t) \gets \mathrm{DP}[t][(c,M)]$
    \STATE $S[t] \gets s_t$
    \STATE $(c,M)\gets (c_{\text{prev}},M_{\text{prev}})$
  \ENDFOR \label{line:alg:witnessing_end}
  \RETURN $S = (s_1,\dots,s_m)$
\end{algorithmic}
\label{alg:DP}
\end{algorithm}

Intuitively, the algorithm explores every possible way to assign the first $t$ goods while precisely tracking the valuation matrix $M$ and cost (i.e., envy) incurred $c$.
At each step, it branches on which agent receives the next good, updates the corresponding column of $M$, accrues the cost incurred, and records a predecessor pointer.
At the end (layer $m$), it scans through all reachable pairs $(c, M)$ and returns the one minimizing cost, reconstructing a witnessing allocation by following the pointers backward.
Our result is as follows.

\begin{theorem} \label{thm:alg_DP}
    Algorithm~\ref{alg:DP} returns an optimal sequence in time $\mathcal{O}( (m+1)(V_{\max}m^2+1)(mV_{\max}+1)^{n^2}\cdot n^3)$.
\end{theorem}
\begin{proof}
    We first prove correctness.
    Observe that the $\{(c,M): (c,M)\in \mathrm{DP}[m]\}$ is the set of pairs attainable after all goods have arrived.
    Then, selecting $(c^*,M^*)$ with minimum $c^*$ (Line~\ref{line:alg:DP_c*M*}) corresponds to one with minimum cost (or cumulative envy). From there, our algorithm is able to reconstruct a corresponding allocation sequence (Lines~\ref{line:alg:witnessing_start}--\ref{line:alg:witnessing_end}).

    Thus, it suffices for us to prove the following invariant.
    Let $\pi_t := (g_1,\dots,g_t)$.

    \medskip
    \noindent\textbf{Invariant.} 
    For every $t\in\{0,\dots,m\}$ and pair $(c,M)$, $\mathrm{DP}[t][(c,M)]$ exists if and only if there is an allocation sequence $X_t$ for the first $t$ goods (prefix $\pi_t$) such that the resulting valuation matrix is $M$ and $c = \sum_{\tau=1}^{t} c_\tau(X_t,\pi_\tau)$.
    Moreover, the predecessor pointer encodes one such sequence.

    We prove the invariant by induction on $t \in \{0,\dots, m\}$.
    
    For the base case, at $t = 0$, $\mathrm{DP}[t]$ only contains $(0,M^0)$, where $M^0$ is the all-zero matrix; and the cost is $0$. 
    Thus the invariant holds trivially.
    
    Next, we prove the inductive step.
    Assume the invariant holds for some $t<m$.
    Since for each round $t$, $c_t(S,\pi)\le t\cdot V_\text{max}$, we have $\sum_{t=1}^m c_t(S,\pi)\le V_\text{max} \cdot m(m+1)/2 < V_\text{max} \cdot m^2$. 
    
    Consider any state $(c,M) \in \mathrm{DP}[t]$.
    By the induction hypothesis, there exists an allocation sequence $X_t = (x_1,\dots,x_t)$ of the first $t$ goods with valuation matrix $M$ and cost $\cost(X_t,\pi_t) = c$.
    When good $g_{t+1}$ arrives, Algorithm~\ref{alg:DP} considers every agent $s\in N$ and for each $s$, constructs $M'$ by increasing column $s$ by $v_i(g_{t+1})$.
    Hence every feasible extension is represented in $\mathrm{DP}[t+1]$.
    
    Conversely, suppose $(c',M')$ is inserted into $\mathrm{DP}[t+1]$ from predecessor $(c,M)$ using agent $s$.
    By the induction hypothesis, $(c,M)$ arises from some allocation sequence $(x_1,\dots,x_t)$ with matrix $M$ and cost $\cost(X_t,\pi_t)$; extending by $x_{t+1}=s$ gives us matrix $M'$ and updating the cost to $c'$, so $(c',M')$ corresponds to a feasible allocation of $t+1$ goods.
    Note that after updating column $s$, $M'$ is exactly the valuation matrix after $t + 1$ items under the extended sequence. 
    Therefore, the maximum envy at round $t+1$ is
    \begin{equation*}
        c_{t+1}=\max_{i,j}\max\{0,\, M'_{i,j}-M'_{i,i}\},
    \end{equation*}
    so the cumulative cost is consistent with Line~\ref{alg:line:cumulativecost} of the algorithm.
    Thus the invariant holds for $t+1$.

    By induction, the invariant holds for all $t\le m$, and optimality follows. 
    
    Next, we prove the running time.
    The state space is indexed by $(t, c, M)$ with $t \in \{0,\dots,m\}, c \in \{0, \dots, V_\text{max} \cdot m^2\}$, and $M \in \{0,\dots, m \cdot V_\text{max}\}^{n\times n}$.
    Thus, the number of states is 
    \begin{align*}
        & (m+1) \cdot (V_\text{max} \cdot m^2 +1) \cdot (m \cdot V_\text{max} + 1)^{n^2}.
    \end{align*}
    From each state, there are $n$ agents the algorithm can consider allocating the incoming good to, updates one column of $M$ and computes the new cost. 
    Both operations (i.e., per transition) take $\mathcal{O}(n^2)$ time.
    Since there are $n$ transitions per state, it takes $\mathcal{O}(n^3)$ per state.
    Thus, the total running time is $\mathcal{O} ( (m+1)(V_\text{max}m^2+1)(mV_\text{max}+1)^{n^2}\cdot n^3 )$. Reconstruction follows $m$ pointers and is $\mathcal{O}(m)$, which does not change the bound.
\end{proof}

This result has two important implications. First, 
for fixed $n$, the running time is polynomial in $m$ and $V_\text{max}$; thus the algorithm is pseudopolynomial.
In particular, if item values are polynomially bounded integers (or given in unary), the algorithm runs in polynomial time. 
\begin{corollary}
    \fixedenvymin{} is \XP{}~\cite{flum_parameterized_2006,niedermeier_invitation_2006} with respect to $(n, V_\text{max})$.
    Equivalently, it is solvable in pseudopolynomial time when $n$ is bounded by a constant.
\end{corollary}

Second, in the case of binary valuations,\footnote{Binary valuations are very well-studied in online fair division \cite{Aleksandrov2015,aleksandrov2017purenashonline,wang2025onlinebivalued} and in the broader fair division literature \cite{suksompong2026char,hosseini2020fairdivisionwithholding,HalpernPrPs20,neoh2025efx,SuksompongTe22}.} the size of the state space in the dynamic program becomes more structured. 
Then, the number of agents is constant, the problem is polynomial-time solvable (i.e. slicewise polynomial in $n$).

\begin{corollary} 
    Under binary valuations, \fixedenvymin{} is \XP{} with respect to $n$.
\end{corollary}

These results show that, despite the general computational hardness of the problem, efficient algorithms are achievable under natural and practically motivated restrictions.

That being said, an intriguing open question remains: is the problem polynomial-time solvable when the number of agents is constant and agents have identical valuation functions? 
We conjecture that the answer is `no'. 

As a partial positive result, we provide an FPTAS for identical valuations (with integer item values) when $n$ is fixed.
\begin{theorem} \label{thm:fptas_identical_nconstant}
    For any fixed $n \geq 2$, under identical valuations with integer-valued items, and for every $\varepsilon \in (0,1)$, there exists a $(1+\varepsilon)$-approximation algorithm whose running time is polynomial in $m$ and $1/\varepsilon$.
\end{theorem}
Our FPTAS follows the standard approach of rounding the values and solving the rounded instance by dynamic programming~\cite{ibarra1975fast,sahni1976algorithms}; Lipton et al.~\cite{lipton2004approximately} use the same approach for the envy-ratio with a constant number of agents. The step specific to our objective is a lower bound of $(1-1/n)\,v(G)$ on the cost of every allocation sequence, which bounds the rounding error added up over all $m$ rounds.

\section{Choosing the Arrival Sequence} \label{sec:variable}
The results presented in Section~\ref{sec:intractability} show that computing a cost-minimizing allocation sequence is generally intractable, even under restricted settings such as identical or binary valuations. 
While Section~\ref{sec:tractability} presents some positive algorithmic results under stronger assumptions, the overall hardness results align with existing impossibility results in the literature on online and temporal fair division: fairness objectives remain elusive in the presence of fixed input sequences.
In fact, we illustrate with an example how an adversary deciding the arrival order can make a significant difference when it comes to minimizing cumulative envy.
\begin{example}\label{ex:epsilon-big}
    Consider an instance with $n=2$ and identical valuation function $v$.
    Fix any integer $q\ge 2$, let $m:=2q+1$, and set
    $\varepsilon:=99/(2q)$. Thus, $m$ is odd and
    $(m-1)\varepsilon=99$. There is one good of value $100$ and
    $m-1$ goods of value $\varepsilon$.
    
    \begin{itemize}
    \item If the arrival order is $\pi=(g_1,\dots,g_m)$ with $v(g_1)=100$
    and $v(g_t)=\varepsilon$ for all $t\ge 2$, then after round $1$, $c_1(S,\pi)=100$ for every allocation sequence $S$.
    Furthermore, each subsequent $\varepsilon$-good can reduce the gap $| v(A_1^t)-v(A_2^t) |$ by at most $\varepsilon$, so for $k=0,1,\dots,m-1$, $c_{1+k}(S,\pi)\ge 100-k\varepsilon$.
    Allocating every $\varepsilon$-good to the currently lagging agent achieves equality at every round, and therefore $\min_{S\in\mathcal{S}} \mathrm{cost}(S,\pi)
    =\sum_{k=0}^{m-1} (100-k\varepsilon)
    =100m-\varepsilon\frac{m(m-1)}{2}
    =\frac{101m}{2}$, which is $\Theta(m)$.
    
    \item In contrast, if the $\varepsilon$-goods arrive first (e.g., $\pi'=(g_2,\dots,g_m,g_1)$),
    then allocating the first $m-1$ goods alternately keeps $c_t(S,\pi')\le \varepsilon$
    for all $t\le m-1$, and the last step contributes $c_m(S,\pi')=100$. Consequently, $\mathrm{cost}(S,\pi')=\frac{m-1}{2}\varepsilon+100=\frac{299}{2}$, which is $O(1)$ in $m$.
    \end{itemize}
    Thus, when one item dominates in value, placing it early can force cumulative cost that is a factor $\Theta(m)$ larger than an order that places it last.
\end{example}
Motivated by this example, we study a variant in which the algorithm chooses
the order of the goods, while the agents' valuations are fixed in advance.
This models settings where the full set of goods is known and goods can be held
briefly before release. For example, a food bank may have several donated
pallets in storage and choose which pallet to allocate next, or a public agency
may choose the order in which approved permits are issued. We ask whether
control over the order can reduce cumulative envy.\footnote{Many related online problems assume a \emph{random arrival order} as a way to circumvent impossibility results that arise under adversarial input models \cite{babaioff2007knapsacksecretary,devanur2009adwords,feldman2009onlineaddisposal,feldman2009stochasticmatching,goel2008onlinebudgeted,hajiaghayi2022santaclaus,karande2011onlinebipartiteunknown,kleinberg2005secretaryonlineauctions}.}
Related forms of flexibility have recently been studied for existence questions: Choi and Li~\cite{choi2026scheduling} allow each good to be delayed by a bounded number of rounds in order to obtain TEF1 allocations, and Amanatidis et al.~\cite{amanatidis2026buffers} let an online algorithm hold and reorder a bounded number of goods. In our variant, the algorithm may choose any order, but each good is still allocated in its own round and the envy at every round counts; hence, choosing the order does not reduce our problem to the offline one.

Importantly, this setting is not inherently ``stronger'' or ``weaker'' than the standard one; it is simply of a different structure. Although the ability to control the arrival order offers additional flexibility, the benchmark for optimization also shifts: the optimal solution under this new setting may differ from that in the fixed-order setting. 
As a result, the optimization problem (e.g., obtaining approximation results) does not necessarily become easier---the performance is still measured relative to a stronger baseline, and the underlying computational challenges remain significant.

Observe that the cost objective is cumulative, so even when the final allocation can be envy-free, early rounds may still incur substantial envy. In the special case of identical valuations, Lemma~\ref{lem:opt-lb-n} (stated below) implies the universal lower bound
$\OPT(\I)\geq (1-1/n) \cdot v(G)$. One might hope that by choosing the arrival order, the algorithm could always
achieve this bound. Our next result shows that even deciding whether this lower bound is attainable
is NP-complete.
\begin{tcolorbox}[title=\variableenvymin]

    Given an instance $\I = (N, G, \mathbf{v})$ and an integer $\gamma \in \mathbb{Z}_{\geq 0}$, can we find a permutation $\pi$ over the goods in $G$ and an allocation sequence $S$ such that $\textsf{cost}(S, \pi) \leq \gamma$?
\end{tcolorbox}

\begin{theorem} \label{thm:variable_hard}
    \variableenvymin{} is \NP-complete even for two agents with identical valuations. 
\end{theorem}

Theorem~\ref{thm:variable_hard} shows that allowing the algorithm to choose the arrival order does not eliminate computational intractability: the problem remains NP-hard even under identical valuations.\footnote{This setting has also been extensively studied both in sequential fair division models \cite{elkind2025temporal,igarashi2023repeatedfairallocation,teh2026onlinehouseallocationsubsidy} and in the broader fair division literature \cite{barman2020identical,choo2024housesubsidy,plaut2020almost,scarlett2023fairness}.}
Nevertheless, this setting admits simple algorithms with provable approximation and additive guarantees.

We first state a universal lower bound (independent of the chosen order), which will serve as a
benchmark for both hardness and approximation results.
Importantly, the bound holds for any number of agents, assuming identical valuations.
Given an instance $\I$, let $\OPT(\I) := \min_{S \in \mathcal{S},\pi \in \Pi} \textsf{cost}(S,\pi)$ denote the minimum possible cost achievable over all arrival orders $\pi$ and allocation sequence $S$.

\begin{lemma} \label{lem:opt-lb-n}
    For any instance $\I=(N,G,\mathbf{v})$ where $n\ge 2$ and agents have identical valuations, $\OPT(\I)\ge (1- 1/n)\, v(G)$.
    Moreover, the bound is tight.
\end{lemma}

Next, we introduce the following algorithm (Algorithm~\ref{alg:greedy_lowest_first}), designed for the setting where agents have identical valuations.
The algorithm proceeds by first sorting the goods in increasing order of value (that will be the $\pi$), then iteratively allocating each good to the agent with the currently smallest bundle value.
Despite its simplicity, we will show that this greedy approach is able to give us several positive approximation guarantees.

\begin{algorithm}[h]
   \caption{\textsc{Identical-Lowest-First}}
   \label{alg:greedy_lowest_first}
\begin{flushleft} 
    {\bf Input}: Instance $\I = (N, G, \mathbf{v})$. \\
    {\bf Output}: An order $\pi$ and allocation sequence  $S = (s_1,\dots,s_m)$.
\end{flushleft}
   \begin{algorithmic}[1]
       \STATE Initialize $\mathcal{A}^{(0)} \leftarrow (\varnothing,\dots,\varnothing)$ and label goods such that $v(g_1) \leq \dots \leq v(g_m)$
       \STATE $\pi \gets (g_1,\dots,g_m)$
       \FOR{$t = 1,\dots,m$}
            \STATE Let $i^* \in \argmin_{i \in N} v(A_i^{t-1})$, with ties broken in favor of the lowest-index agent
            \STATE $s_t \leftarrow i^*$
           \STATE $A_{i^*}^{t}\gets A_{i^*}^{t-1}\cup \{g_t\}$\\
           \STATE $A_i^t \gets A_i^{t-1}$ for all $i \in N \setminus \{i^*\}$
       \ENDFOR
       \STATE \textbf{return} $\pi$ and allocation sequence $S=(s_1,\dots,s_m)$
   \end{algorithmic}
\end{algorithm}
\paragraph{Relation to envy-cycle elimination and list scheduling.}
Under identical valuations, no envy cycle can arise, and an agent is unenvied exactly when her bundle has minimum value. Hence, Algorithm~\ref{alg:greedy_lowest_first} is the envy-cycle elimination algorithm of Lipton et al.~\cite{lipton2004approximately} applied to the goods in nondecreasing order of value. Their analysis, applied to the first $t$ goods, shows that after each round the maximum envy is at most the largest value of a good allocated so far, i.e., $c_t \le v(g_t)$; our proofs combine this bound with the lower bound of Lemma~\ref{lem:opt-lb-n} and with an exact expression for the cost of Algorithm~\ref{alg:greedy_lowest_first}. Note that Algorithm~\ref{alg:greedy_lowest_first} uses the opposite order to Graham's greedy rule~\cite{graham1969bounds}, which processes goods in nonincreasing order of value and, under identical valuations, yields EFX allocations~\cite{barman2020approximation,plaut2020almost}; placing large goods last keeps the envy small in early rounds (cf.\ Example~\ref{ex:epsilon-big}).
Lipton et al.~\cite{lipton2004approximately} also observed that, under identical valuations, envy minimization is closely related to scheduling on identical machines. Indeed, Algorithm~\ref{alg:greedy_lowest_first} has the same steps as shortest-processing-time-first list scheduling on identical machines: goods
are jobs, agents are machines, and $v(g)$ is the processing time of job $g$. The objectives are different, however. List scheduling is usually analyzed
using the final makespan, that is, the largest final machine load. In our
setting, $c_t$ is the difference between the largest and smallest loads after
the first $t$ jobs, and we minimize $\sum_t c_t$. Thus, the $3/2$ factor in
Theorem~\ref{thm:alg_3/2approx} does not follow from the usual list-scheduling analysis, although
the common greedy rule may help explain why the same constant appears
\cite{graham1966bounds}.

First, we show that in the case of two agents, we can get a $3/2$-approximation to the optimal cost.
\begin{theorem} \label{thm:alg_3/2approx}
    When $n=2$ and the agents have identical valuations, Algorithm~\ref{alg:greedy_lowest_first} achieves a $3/2$-approximation to the optimal cost.
\end{theorem}

We can extend the analysis to any $n$, obtaining an $\frac{n}{n-1}$-approximation (which approaches $1$ as $n$ grows).
\begin{theorem} \label{thm:alg_2approx_nagents}
    For all $n \geq 2$,  if the agents have identical valuations, then Algorithm~\ref{alg:greedy_lowest_first} achieves a $\frac{n}{n-1}$-approximation to the optimal cost.
\end{theorem}
For $n=2$, the $3/2$-approximation guarantee of Theorem~\ref{thm:alg_3/2approx} is tight; see Appendix~\ref{app:eg_tightness_n=2}.
For $n\ge 3$, Appendix~\ref{app:eg_tightness_ngeq3} provides an instance on which \textsc{Identical-Lowest-First} achieves approximation ratio $\frac{2n-1}{2n-2}$.

Finally, we complement the multiplicative bounds with an additive approximation guarantee, stated in terms of
$V_\text{max}$.

\begin{theorem} \label{thm:alg_1/2approx_additive}
    For all $n \geq 2$, if agents have identical valuations, then Algorithm~\ref{alg:greedy_lowest_first} outputs an order $\pi$ and allocation sequence $S$ such that $\cost(S,\pi) \le \OPT(\I) + \left(1-\frac{1}{n}\right)\,V_{\max}$.
\end{theorem}
In Appendix~\ref{app:eg_tightness_additive}, we present a matching lower bound.

\section{Conclusion}
In this paper, we initiated the study of envy minimization in temporal fair division from an optimization perspective, departing from the traditional focus on discrete relaxations like EF1 and EFX and on the prefix-based notions studied in earlier temporal work.
In the standard model where the order of goods is adversarially chosen and fixed, we show that minimizing this objective is computationally intractable and cannot be approximated within any constant factor unless $\P=\NP$.
In spite of this, we provide a dynamic programming algorithm that computes an optimal solution in several special cases.
We then explored a natural variant in which the algorithm is allowed to choose the arrival order of goods, while an adversary pre-selects the valuations. Even with this additional flexibility, we show that the problem remains \NP-complete for just two agents and identical valuations.
We complement this with efficient approximation algorithms.

Our results contribute to the growing understanding of temporal fairness in sequential allocation problems. Several directions remain open. In particular, it would be interesting to study discounted cumulative envy where recent unfairness matters more than past inequities; Kahana et al.~\cite{kahana2026perpetual} recently studied a related discounted model in an online setting.

\subsection*{Acknowledgments} PWG was supported by EPSRC grants EP/X040461/1 and EP/X038548/1. PWG and NT were supported by ARIA project ``Aggregating Safety Preferences for AI Systems: A Social Choice Approach''. IR was supported by a Rhodes scholarship.

\bibliographystyle{plain}
\bibliography{bib}

\newpage

\appendix

\begin{center}
    \Large\bf
    Appendix
\end{center}
\vspace{2mm}

\section{Omitted Proofs from Section~\ref{sec:intractability}}

\subsection{Proof of Theorem~\ref{thm:intractable_binary}}
    We reduce from the \textsc{Binary-EF} problem, which is known to be NP-hard \cite{aziz2015fairassignmentordinal,hosseini2020fairdivisionwithholding}.
    An instance of \textsc{Binary-EF} is given by a set of agents $\widehat{N} = \{1,\dots,\widehat{n}\}$, a set of goods $\widehat{G} = \{\widehat{g}_1,\dots,\widehat{g}_{\widehat{m}}\}$, and a valuation profile $\widehat{\mathbf{v}} = (\widehat{v}_1,\dots, \widehat{v}_{\widehat{n}})$ where agents have additive binary valuations over the set of goods; it is a yes-instance if there exists an allocation $\widehat{\mathcal{A}}$ of goods to agents such that $\widehat{\mathcal{A}}$ is envy-free, i.e., $\widehat{v}_i(\widehat{A}_i) \geq \widehat{v}_i(\widehat{A}_j)$ for every pair of agents $i, j \in \widehat{N}$, and a no-instance otherwise. 

    Construct an instance of our problem, where $N = [\widehat{n}]$ and we have $\widehat{m} + \widehat{m}^3$ goods arriving in the order $\pi = (g_1,\dots, g_{\widehat{m} + \widehat{m}^3})$, with valuations defined as follows:
    \begin{equation*}
        v_i(g_j) = 
        \begin{cases}
            \widehat{v}_i(\widehat{g}_j) & \text{ if } j \in \{1,\dots, \widehat{m} \}; \\
            0 & \text{ otherwise}.
        \end{cases}
    \end{equation*}
    We first show that \fixedenvymin{} is NP-complete (it is easy to see that the problem is in NP, we thus focus on proving NP-hardness).
    We do so by proving that there exists an envy-free allocation $\widehat{\mathcal{A}}$ for \textsc{Binary-EF} if and only if there exists an allocation $\mathcal{A}$ (and hence corresponding allocation sequence $S$) such that $\textsf{cost}(S,\pi) \leq \widehat{m}^3-1$ for the above constructed instance.

    The zero-valued goods make any envy present after the first $\widehat m$ rounds persist. If the \textsc{Binary-EF} instance is a yes-instance, we allocate the first $\widehat m$ goods according to an envy-free allocation,
so all later rounds have zero envy. If it is a no-instance, every allocation
has maximum envy at least $1$ at round $\widehat m$, and this remains true
while the zero-valued goods arrive. The same idea, with additional zero-valued
goods, gives the constant-factor inapproximability result.

    ($\Rightarrow$) Suppose there exists an envy-free allocation $\widehat{\mathcal{A}}$ for \textsc{Binary-EF}.
    Then, consider the sequence $S = (s_1,\dots, s_{\widehat{m} + \widehat{m}^3})$ whereby
    \begin{equation*}
        s_j = 
        \begin{cases}
            \ell & \text{ if } j \leq \widehat{m} \text{ and } \widehat{g}_j \in \widehat{A}_\ell, \text{ for } \ell \in [\widehat{n}]; \\
            1 & \text{ if } j > \widehat{m}.
        \end{cases}
    \end{equation*}
    Observe that for each $t \in [\widehat{m}-1]$,\footnote{This is because under binary valuations, at round $t$, each bundle contains at most $t$ goods and each good is worth at most $1$ to any agent.} $c_t(S,\pi) \leq t$, and $c_{\widehat{m}}(S,\pi) = 0$, since $\widehat{\mathcal{A}}$ is envy-free, which means under $\mathcal{A}^{\widehat{m}}$, there is no envy (and cost is $0$).
    Thus, we also get that for $t' \in \{\widehat{m} + 1,\dots, \widehat{m} + \widehat{m}^3\}$, $c_{t'}(S,\pi) = 0$, giving us
    \begin{align*}
        \textsf{cost}(S, \pi) & =  \sum_{t=1}^{\widehat{m} + \widehat{m}^2} c_t(S,\pi) \\
        & = \sum_{t=1}^{\widehat{m} - 1} c_t(S,\pi) + c_{\widehat{m}}(S,\pi) + \sum_{t = \widehat{m}+1}^{\widehat{m} + \widehat{m}^3} c_t(S,\pi) \\
        & \leq \sum_{t=1}^{\widehat{m}-1} c_t(S,\pi) + 0 + 0= \frac{\widehat{m}(\widehat{m}-1)}{2}.
    \end{align*}
    Then, $\frac{\widehat{m}(\widehat{m}-1)}{2} \leq \widehat{m}^3-1$ for all $\widehat{m} \geq 1$.

    ($\Leftarrow$) Suppose there does not exist an envy-free allocation $\widehat{\mathcal{A}}$ for \textsc{Binary-EF}.
    This means that for any sequence $S$, the allocation induced on the first $\widehat{m}$ goods is not envy-free, and $c_{\widehat{m}}(S,\pi) \geq 1$. Since the next $\widehat{m}^3$ goods $g_{\widehat{m}+1}, \dots, g_{\widehat{m}+ \widehat{m}^3}$ are zero-valued, we have that $c_{\widehat{m}+1}(S,\pi) = \dots = c_{\widehat{m} + \widehat{m}^3}(S,\pi) \geq 1$. 
    Thus, for any sequence $S$,
    \begin{equation*}
        \textsf{cost}(S, \pi) \geq  \sum_{t=\widehat{m}}^{\widehat{m}+\widehat{m}^3} c_t(S,\pi) \geq \widehat{m}^3 + 1 > \widehat{m}^3-1.
    \end{equation*}
    
    Next, we argue that in (the same) setting where agents have binary valuations, there exists no constant-factor approximation algorithm, unless P = NP.

    Note that in the reduction above, we showed that given an instance $\I$ and order $\pi$, determining whether we can find an allocation sequence $S$ such that $\cost(S,\pi) \leq \widehat{m}^3 - 1$ is NP-complete.
    Equivalently, 
    \begin{itemize}
        \item any yes-instance admits a sequence $S$ with $\cost(S,\pi) \leq \widehat{m}^3 - 1$; 
        \item any no-instance will force any sequence $S$ to have $\cost(S,\pi) \geq \widehat{m}^3 + 1$.
    \end{itemize}

    Observe that for yes-instances, $c_{\widehat{m}}(S,\pi) = 0$, which means the cost of $\widehat{m}^3-1$ is incurred at earlier rounds, whereas for no-instances, $c_{\widehat{m}}(S, \pi) \geq 1$.

    Let $B:=m_b^3-1$ and fix any constant $\alpha\ge 1$. We may assume $m_b\ge 2$, so $B>0$. Append
\[
    Z_\alpha:=\left\lceil B(\alpha-1)\right\rceil+1
\]
additional zero-valued goods. If the Binary-EF instance is a yes-instance, the same allocation still has cost at most $B$. If it is a no-instance, then envy is at least $1$ from round $m_b$ onward, and therefore it remains at least $1$ throughout the appended zero-valued rounds. Thus every no-instance has cost at least
\[
    m_b^3+1+Z_\alpha > \alpha B .
\]
Thus, running an $\alpha$-approximation on the constructed
\textsc{Temporal-EnvyMin} instance and comparing its returned cost with
$\alpha B$ would decide the original \textsc{Binary-EF} instance in
polynomial time. This is impossible unless $\P = \NP$.

\section{Omitted Proofs from Section~\ref{sec:tractability}}

\subsection{Proof of Proposition~\ref{prop:tractability_binaryidentical}}

    Consider the greedy algorithm that, at each round $t \in [m]$, allocates the arriving good to an agent with minimum current bundle value $v(A_i^{t-1})$.
    This can be implemented in $\mathcal{O}(mn)$ time.

    Since valuations are binary and identical, each bundle value $v(A_i^t)$ is an integer that represents the number of goods of value $1$ assigned to agent~$i$ among the first $t$ arrivals.
    Let $q_t$ be the number of goods of value $1$ among the first $t$ goods in $\pi$.
    Under identical valuations,
    \begin{align*}
        c_t(S,\pi) = \max_{i,j \in N} \left( v(A_j^t) - v(A_i^t) \right) = \max_{i \in N} v(A_i^t) - \min_{i \in N} v(A_i^t).
    \end{align*}
    Then, since $\sum_{i \in N} v(A_i^t) = q_t$ and all $v(A_i^t)$ are integers,
    \begin{equation*}
        \max_{i \in N} v(A_i^t) \geq \lceil q_t / n\rceil \text{ and } \min_{i \in N} v(A_i^t) \leq \lfloor q_t /n\rfloor.
    \end{equation*}
    Thus, for every allocation,
    \begin{equation*}
        c_t(S,\pi) \geq \lceil q_t/n \rceil - \lfloor q_t/n \rfloor,
    \end{equation*}
    which equals $0$ if and only if $q_t \equiv 0 \pmod n$, and equals $1$ otherwise.

    We claim the greedy algorithm maintains
    \begin{equation*}
        \max_{i \in N} v(A_i^t) - \min_{i \in N} v(A_i^t) \leq 1 \text{ for all } t,
    \end{equation*}
    proved by induction on $t$: initially it is $0$.
    If the arriving good has value $0$, bundle values do not change.
    If it has value $1$, the algorithm adds $1$ to a minimum bundle value, so the gap cannot increase beyond $1$.
    Thus, $c_t(S,\pi)\leq 1$ always. 
    Moreover, if $q_t \not\equiv 0 \pmod n$, the bundle values cannot be equal integers summing to $q_t$, so $c_t(S,\pi)\neq 0$ and thus $c_t(S,\pi)= 1$.
    If $q_t \equiv 0 \pmod n$, the gap $\leq 1$ forces equality, giving us $c_t(S,\pi)= 0$.
    To see why, if $\max_{i \in N} v(A_i^t)-\min_{i \in N} v(A_i^t)\le 1$, then all $v(A_i^t)\in\{ \min_{i \in N} v(A_i^t), \min_{i \in N} v(A_i^t)+1\}$. Writing $r$ for the number of agents with value $\min_{i \in N} v(A_i^t)+1$, we have $q_t=n \min_{i \in N} v(A_i^t) +r$. If $q_t\equiv 0\pmod n$, then $r\equiv 0\pmod n$, but $0\le r\le n-1$, thus $r=0$ and all values equal.
    
    Thus, the greedy algorithm achieves a per-round lower bound for every $t$, and summing over $t$ shows that it minimizes $\sum_{t=1}^m c_t(S,\pi)$.

\subsection{Proof of Theorem~\ref{thm:fptas_identical_nconstant}}
        Fix a constant $n\ge 2$, an instance $\I=(N,G,v)$ with identical valuation function $v$, and a fixed arrival order $\pi=(g_1,\dots,g_m)$. 
    Let $\OPT := \min_{S\in\mathcal{S}} \cost_v(S,\pi)$ denote the optimal cumulative maximum envy for this fixed order under $v$.

    If $v(G)=0$, then every bundle has value $0$ at every round, so $c_t(S,\pi)=0$ for all $t$ and all $S$. Hence $\mathrm{OPT}=0$, and any allocation sequence is optimal. We therefore assume $v(G)>0$.
    Set
    \begin{equation*}
        \delta := \frac{\varepsilon (1-1/n) \cdot v(G)}{m(m+1)},
    \end{equation*}
    and define a scaled (integer) valuation $\tilde{v}$ by
    \begin{equation*}
        \tilde v(g) := \lfloor v(g)/\delta\rfloor \quad \text{for each } g\in G.
    \end{equation*}

    Let $\tilde{\I}:=(N,G,\tilde{v})$ be the scaled instance. 
    Compute an allocation sequence $\tilde{S}$ that optimizes $\cost_{\tilde v}(S,\pi)$ for the instance $\tilde{\I}$ and the same order $\pi$ (e.g., using Algorithm~\ref{alg:DP} together with Theorem~\ref{thm:alg_DP}), and output $\tilde S$ for the original instance.

    For every $g\in G$,
    \begin{equation*}
        \tilde v(g) \le \frac{v(g)}{\delta} \le \frac{v(G)}{\delta} = \frac{m(m+1)}{\varepsilon(1-1/n)},
    \end{equation*}
    so $\tilde V_\text{max}:=\max_{g\in G}\tilde v(g)=\mathcal{O}(m^2/\varepsilon)$ (for fixed $n$). By Theorem~\ref{thm:alg_DP}, for fixed $n$, Algorithm~\ref{alg:DP} runs in time polynomial in $m$ and $\tilde V_\text{max}$, hence polynomial in $m$ and $1/\varepsilon$.

    Now, we prove the approximation guarantee.
    For any allocation sequence $S$ and any round $t$, write
    \begin{equation*}
        x_i^t := v(A_i^t)\quad\text{and}\quad \tilde x_i^t := \tilde v(A_i^t)\quad \text{for each } i\in N.
    \end{equation*}
    Since $\delta\tilde v(g)\le v(g) < \delta(\tilde v(g)+1)$ for each good $g$ and $|A_i^t| \leq t$ (because at most $t$ goods have arrived by round $t$), we have that for every $i\in N$
    and $t\in[m]$,
    \begin{equation}\label{eq:bundle-sandwich}
        \delta\,\tilde x_i^t \ \le\ x_i^t \ <\ \delta\,\tilde x_i^t + t\delta.
    \end{equation}
    Define $U_t:=\max_{i\in N}x_i^t$ and $L_t:=\min_{i\in N}x_i^t$, and similarly
    $\tilde U_t:=\max_{i\in N}\tilde x_i^t$ and $\tilde L_t:=\min_{i\in N}\tilde x_i^t$.
    Under identical valuations we have
    \begin{equation*}
        c_t(S,\pi)=U_t-L_t \quad\text{and}\quad \tilde c_t(S,\pi)=\tilde U_t-\tilde L_t.
    \end{equation*}
    From \eqref{eq:bundle-sandwich}, we get $U_t\le \delta\tilde U_t+t\delta$ and $L_t\ge \delta\tilde L_t$, thus
    \begin{equation*}
        c_t(S,\pi)=U_t-L_t \le \delta(\tilde U_t-\tilde L_t)+t\delta =\delta\tilde c_t(S,\pi)+t\delta.
    \end{equation*}
    Also $U_t\ge \delta\tilde U_t$ and $L_t\le \delta\tilde L_t+t\delta$, hence
    \begin{equation*}
        c_t(S,\pi)=U_t-L_t  \ge \delta(\tilde U_t-\tilde L_t)-t\delta = \delta \tilde c_t(S,\pi)-t\delta.
    \end{equation*}

    Therefore, for every $S$,
    \begin{align}\label{eq:cost-sandwich}
        \delta\cost_{\tilde v}(S,\pi)-\delta\sum_{t=1}^m t
         & \le \cost_v(S,\pi) \nonumber \\
         & \le
        \delta \cost_{\tilde v}(S,\pi)+\delta\sum_{t=1}^m t.
    \end{align}
    Since $\sum_{t=1}^m t = m(m+1)/2$, \eqref{eq:cost-sandwich} becomes
    \begin{align}\label{eq:cost-sandwich-closed}
        & \delta \cost_{\tilde v}(S,\pi)-\delta\frac{m(m+1)}{2} \nonumber \\
         & \le \cost_v(S,\pi) \nonumber \\ 
         & \le \delta \cost_{\tilde v}(S,\pi)+\delta\frac{m(m+1)}{2}.
    \end{align}

    Let $S^*$ be optimal for the original instance (so $\cost_v(S^*,\pi)=\OPT$). 
    By optimality of $\tilde{S}$ for the scaled instance, $\cost_{\tilde v}(\tilde S,\pi)\le \cost_{\tilde v}(S^*,\pi)$. Applying \eqref{eq:cost-sandwich-closed} first to $\tilde S$ and then to $S^*$ gives us
    \begin{align*}
        \cost_v(\tilde S,\pi) &
        \le \delta\cost_{\tilde v}(\tilde S,\pi)+\delta\frac{m(m+1)}{2} \\
        & \le \delta\cost_{\tilde v}(S^*,\pi)+\delta\frac{m(m+1)}{2} \\
        & \le \cost_v(S^*,\pi)+\delta m(m+1) \\
        & = \OPT + \delta m(m+1).
    \end{align*} 
    
    It remains to lower bound $\OPT$ in terms of $v(G)$. Fix any sequence $S$ and define $x_i^t,U_t,L_t$ as above.
    If $g_t$ is allocated to agent $j$, then $x_j^t-x_j^{t-1}=v(g_t)$. Since $U_t\ge x_j^t$ and $L_{t-1}\le x_j^{t-1}$,
    \begin{equation*}
        v(g_t)\ =\ x_j^t-x_j^{t-1}\ \le\ U_t-L_{t-1}.
    \end{equation*}
    Using $U_t=(U_t-L_t)+L_t=c_t(S,\pi)+L_t$, we get
    \begin{equation*}
        v(g_t) \le c_t(S,\pi) + (L_t-L_{t-1}).
    \end{equation*}
    Summing over $t=1,\dots,m$ and using $L_0=0$ gives us
    \begin{equation*}
        v(G) = \sum_{t=1}^m v(g_t) \le \sum_{t=1}^m c_t(S,\pi) + L_m
        = \cost_v(S,\pi)+L_m.
    \end{equation*}
    Since $L_m=\min_{i\in N}v(A_i^m)\le \frac1n\sum_{i\in N}v(A_i^m)=v(G)/n$, it follows that
    $\cost_v(S,\pi)\ge (1-\tfrac1n)v(G)$. Minimizing over $S$ yields
    \begin{equation}\label{eq:opt-lb}
    \OPT \ge (1-\tfrac1n)\,v(G).
    \end{equation}
    Finally, by our choice of $\delta$,
    \begin{equation*}
         \delta m(m+1)=\varepsilon(1-\tfrac1n)v(G) \le \varepsilon \OPT,
    \end{equation*}
    where the inequality follows from \eqref{eq:opt-lb}. 
    Consequently,
    \begin{equation*}
        \cost_v(\tilde S,\pi) \le \OPT+\varepsilon\OPT = (1+\varepsilon)\OPT,
    \end{equation*}
    which proves the claimed $(1+\varepsilon)$-approximation and completes the proof.

\section{Omitted Proofs from Section~\ref{sec:variable}}
\subsection{Proof of Theorem~\ref{thm:variable_hard}}
    Membership in NP is easy to see: given $\pi=(g_1,\dots,g_m)$ and $S=(s_1,\dots,s_m)$, we can compute all partial bundles $(A_i^t)_{i\in N,t\in[m]}$ and the values $c_t(S,\pi)$ for each $t \in [m]$, and verify whether $\sum_{t=1}^m c_t(S,\pi)\le \gamma$ in polynomial time.
    Thus, we focus on showing NP-hardness.
    
    We reduce from the NP-hard \textsc{Partition} problem~\cite{garey1979computers}.
    An instance of \textsc{Partition} consists of
    positive integers $x_1,\dots,x_q$ with total sum $T:=\sum_{i=1}^q x_i$; we may assume $T$ is even (otherwise it is a trivial no-instance).  We construct an instance $\I=(N,G,\mathbf{v})$ of \textsc{Temporal-EnvyMin-Seq} with $N = \{1,2\}$ and identical additive valuations (so $v_1=v_2=v$).
    Let $M := T+1$.
    Create $2q$ goods (thus the number of goods is $m = |G| = 2q$), consisting of pairs $(a_i,b_i)$ for each $i\in[q]$, where
    \begin{equation*}
        v(a_i) = M \quad \text{and} \quad v(b_i)=M+x_i.
    \end{equation*}
    Then the total value is
    \begin{equation*}
        v(G)=\sum_{i=1}^q (v(a_i)+v(b_i) )=\sum_{i=1}^q (2M+x_i)=2qM+T,
    \end{equation*}
    and let
    \begin{equation*}
        \gamma := \frac{v(G)}{2}=qM+\frac{T}{2},
    \end{equation*}
    which is an integer since $T$ is even.
    
    We will use the following fact, which holds for every order $\pi$ and sequence $S$ in this constructed instance.
    Since $n=2$ and valuations are identical, for every $t\in[m]$ we have
    \begin{equation*}
        c_t(S,\pi)=|v(A_1^t)-v(A_2^t)|.
    \end{equation*}
    Define $D_t:=v(A_1^t)-v(A_2^t)$, so $c_t(S,\pi)=|D_t|$.  For each $t\ge 2$, the $t$-th arriving good $g_t$ is allocated to exactly one agent, so $D_t=D_{t-1}\pm v(g_t)$ and therefore $|D_t-D_{t-1}|=v(g_t)$.  
    By the triangle inequality, we have that
    \begin{equation}\label{eq:triangle-step}
        v(g_t)=|D_t-D_{t-1}|\le |D_t|+|D_{t-1}|=c_t(S,\pi)+c_{t-1}(S,\pi)
    \end{equation}
    for all $t\ge 2$.
    Summing~\eqref{eq:triangle-step} over $t=2,\dots,m$ we get
    \begin{equation*}
        \sum_{t=2}^m (c_{t-1}(S,\pi)+c_t(S,\pi)) \ge \sum_{t=2}^m v(g_t)
        = v(G)-v(g_1).
    \end{equation*}
    The left-hand side equals $c_1(S,\pi)+c_m(S,\pi)+2\sum_{t=2}^{m-1} c_t(S,\pi)
    =2\sum_{t=1}^m c_t(S,\pi)-c_1(S,\pi)-c_m(S,\pi)$, so
    \begin{equation*}
        2 \cost(S,\pi)-c_1(S,\pi)-c_m(S,\pi) \ge v(G)-v(g_1).    
    \end{equation*}
    Moreover, after the first good arrives exactly one agent has received $g_1$ and the other has an empty bundle, thus $c_1(S,\pi)=v(g_1)$.  
    Substituting gives us
    \begin{equation*}
        2 \cost(S,\pi)-c_m(S,\pi) \ge v(G).
    \end{equation*}
    Since $c_m(S,\pi)\ge 0$, we conclude that for every $(S,\pi)$,
    \begin{equation}\label{eq:global-lb}
        \cost(S,\pi) \ge \frac{v(G)}{2} = \gamma.
    \end{equation}
    We now prove correctness of the reduction.
    
    ($\Rightarrow$) Suppose the \textsc{Partition} instance is a yes-instance, so there exists
    $P\subseteq[q]$ with $\sum_{i\in P} x_i=T/2$.  We define a permutation $\pi$ of the $2q$ goods by
    listing the goods in $q$ consecutive pairs as follows:
    \begin{itemize}
        \item for each $i\in P$, place the pair $(a_i,b_i)$ in $\pi$, 
        \item for each $i\notin P$, place the pair $(b_i,a_i)$ in $\pi$,
    \end{itemize}
    with all pairs corresponding to indices $i \in P$ appearing first (in arbitrary order among themselves), followed by all pairs corresponding to indices $i \notin P$ (in arbitrary order among themselves).
    
    Next define an allocation sequence $S$ by the greedy rule: at each round $t$, allocate $g_t$ to an agent in $\arg\min_{i\in[2]} v(A_i^{t-1})$ (break ties arbitrarily).
    We analyze the contribution of each consecutive pair to the total cost.
    Fix a pair and let $d$ be the absolute difference $|v(A_1^{t-1})-v(A_2^{t-1})|$ immediately before the first good of the pair is allocated.
    
    We first consider pairs $(a_i,b_i)$ with $i\in P$. 
    During this part of the sequence, the difference $d$ is the sum of the $x$-values of some already-processed indices in $P$, hence $d\le \sum_{i\in P} x_i=T/2<M$.
    Assume w.l.o.g. that $v(A_1^{t-1})\le v(A_2^{t-1})$, so $v(A_2^{t-1})-v(A_1^{t-1})=d$.
    The greedy rule gives $a_i$ to agent~1, so the new difference is $M-d$ and the incurred cost at that round is $M-d$.
    Now agent~2 has the smaller bundle, so $b_i$ is given to agent~2; the new difference becomes
    \begin{equation*}
        (M+x_i)-(M-d)=d+x_i,
    \end{equation*}
    and the cost at the next round is $d+x_i$.
    Therefore, this pair contributes $(M-d)+(d+x_i)=M+x_i$ to $\cost(S,\pi)$, and the absolute difference after the pair increases by $x_i$.
    
    After all $i\in P$ are processed, the absolute difference equals $\sum_{i\in P} x_i=T/2$.
    
    Next, we consider pairs $(b_i,a_i)$ with $i\notin P$.
    At the start of this part, the difference is $T/2$, and we maintain the invariant that after processing some subset $Q\subseteq([q]\setminus P)$, the
    current absolute difference is $\sum_{i\in([q]\setminus P)\setminus Q} x_i$.  In particular, before processing any remaining $i\notin P$, we have $d\ge x_i$.
    Again assume w.l.o.g. that agent~1 has the smaller bundle before the pair, so agent~2 is ahead by $d$.
    The greedy rule assigns $b_i$ to agent~1.  Since $M+x_i>d$, the new difference is $(M+x_i)-d$ and the cost at that round is $M+x_i-d$.
    Then the greedy rule assigns $a_i$ to agent~2, and the new difference becomes
    \begin{equation*}
        \bigl|(M+x_i-d)-M\bigr|=|x_i-d|=d-x_i,
    \end{equation*}
    using $d\ge x_i$.  
    The pair therefore contributes
    \begin{equation*}
        (M+x_i-d)+(d-x_i)=M
    \end{equation*}
    to $\cost(S,\pi)$ and decreases the absolute difference by $x_i$, preserving the invariant.
    After all $i\notin P$ are processed, the absolute difference returns to $0$.
    
    Summing contributions over all $q$ pairs,
    \begin{align*}
        \cost(S,\pi)& =\sum_{i\in P} (M+x_i) + \sum_{i\notin P} M \\
        & = qM+\sum_{i\in P} x_i = qM+\frac{T}{2} =\gamma.
    \end{align*}
    Thus, there exists $(S,\pi)$ with $\cost(S,\pi)\le \gamma$.
    
    ($\Leftarrow$) Conversely, suppose there exist $\pi$ and $S$ with $\cost(S,\pi)\le \gamma$.
    By~\eqref{eq:global-lb}, we must have $\cost(S,\pi)=\gamma=v(G)/2$.
    From
    $2 \cost(S,\pi)-c_m(S,\pi)\ge v(G)$ and $\cost(S,\pi)=v(G)/2$, we get $c_m(S,\pi)=0$.
    Also, define for each $t\ge 2$ the slack
    $\delta_t := \bigl(c_{t-1}(S,\pi)+c_t(S,\pi)\bigr)-v(g_t)\ge 0$ from~\eqref{eq:triangle-step}.
    Since equality holds in the summed inequality, we must have $\sum_{t=2}^m \delta_t=0$, and therefore $\delta_t=0$ for every $t\ge 2$, i.e., for all $t=2,\dots,m$,
    \begin{equation}\label{eq:pair-eq}
        c_{t-1}(S,\pi)+c_t(S,\pi)=v(g_t).
    \end{equation}
    Rearranging gives us $c_t(S,\pi)=v(g_t)-c_{t-1}(S,\pi)$ for all $t\ge 2$, and unrolling this recurrence, we get
    \begin{align*}
        & c_m(S,\pi) \\
        & = v(g_m)-v(g_{m-1})+v(g_{m-2})-\cdots+(-1)^{m-1} v(g_1).
    \end{align*}
    Since $m=2q$ is even and $c_m(S,\pi)=0$, we get
    \begin{equation*}
        \sum_{t\in[m], t\text{ is even}} v(g_t) = \sum_{t\in[m], t\text{ is odd}} v(g_t).
    \end{equation*}
    Every good in $G$ has value either $M$ or $M+x_i$.  Write $v(g_t)=M+y_t$, where
    $y_t\in\{0,x_1,\dots,x_q\}$ and each $x_i$ appears exactly once among the $y_t$'s (coming from the
    unique good $b_i$ of value $M+x_i$).  Since there are exactly $q$ even positions and $q$ odd positions, the $M$-terms cancel, giving us
    \begin{equation*}
        \sum_{t\in[m], t\text{ is even}} y_t = \sum_{t\in[m], t\text{ is odd}} y_t.
    \end{equation*}
    Because $\sum_{t=1}^m y_t=\sum_{i=1}^q x_i=T$, each side must equal $T/2$.
    Let $P\subseteq[q]$ be the set of indices $i$ such that the good $b_i$ (the unique good of value $M+x_i$) appears in an even position of $\pi$.
    Then
    \begin{equation*}
        \sum_{i\in P} x_i = \sum_{t\in[m], t\text{ is even}} y_t = \frac{T}{2},
    \end{equation*}
    so the original \textsc{Partition} instance is a yes-instance.
    
    This completes the reduction, and proves our result.

\subsection{Proof of Lemma~\ref{lem:opt-lb-n}}
    Fix any arrival order $\pi=(g_1,\dots,g_m)$ and allocation sequence $S$.
    For each $t\in\{0,1,\dots,m\}$ define
    \[
    x_i^t := v(A_i^t),\quad
    L_t := \min_{i\in N} x_i^t,\quad
    U_t := \max_{i\in N} x_i^t.
    \]
    Under identical valuations,
    \[
    c_t(S,\pi)
    = \max_{i,j\in N}\max\{0, x_j^t-x_i^t\}
    = U_t-L_t.
    \]
    Now fix any $t\in\{1,\dots,m\}$ and let $j$ be the agent who receives $g_t$.
    By additivity, $x_j^t-x_j^{t-1}=v(g_t)$.
    Since $U_t\ge x_j^t$ and $L_{t-1}\le x_j^{t-1}$, we have
    \[
    v(g_t)=x_j^t-x_j^{t-1}\ \le\ U_t-L_{t-1}.
    \]
    Using $U_t = (U_t-L_t)+L_t = c_t(S,\pi)+L_t$, this becomes
    \[
    v(g_t)\ \le\ c_t(S,\pi) + (L_t-L_{t-1}).
    \]
    Summing over $t=1,\dots,m$ and using $L_0=\min_i v(A_i^0)=v(\varnothing)=0$, we get
    \begin{align*}
        v(G)=\sum_{t=1}^m v(g_t) & \le \sum_{t=1}^m c_t(S,\pi) + \sum_{t=1}^m (L_t-L_{t-1}) \\
        & = \cost(S,\pi) + L_m.
    \end{align*}
    Rearranging gives us
    \begin{equation*}
        \cost(S,\pi)\ge v(G)-L_m = v(G)-\min_i v(A_i^m).
    \end{equation*}
    
    Finally, 
    \begin{equation*}
        \min_i v(A_i^m)\le \frac1n\sum_{i\in N} v(A_i^m)=\frac{v(G)}{n},
    \end{equation*}
    so $\cost(S,\pi)\ge (1-\frac1n)v(G)$.
    Taking a minimum over $(S,\pi)$ gives us the same bound for $\OPT(\I)$.

    Next, we prove the factor $1-1/n$ is tight.
    For any $k\ge 1$, take $m=kn$ goods each with value $1$, so $v(G)=kn$.
    Consider the round-robin allocation (one good per agent each cycle).
    At every round $t$ that is not a multiple of $n$, the bundle values differ by exactly $1$,
    so $c_t=1$; at rounds $t$ that are multiples of $n$, all bundle values are equal and $c_t=0$.
    Thus $\cost(S,\pi) = kn-k = k(n-1) = (1-\frac1n)\cdot v(G)$, matching the lower bound.

\subsection{Proof of Theorem~\ref{thm:alg_3/2approx}}
    Consider any instance $\I = (N, G, \mathbf{v})$ where $N = \{1,2\}$ and agents have identical valuations.
    Let $S = (s_1,\dots,s_m)$ be the allocation sequence induced from the order $\pi$ and allocation $\mathcal{A}$ returned by the algorithm.
    
    Also consider any $t \in \{2,\dots,m\}$.
    Suppose, without loss of generality, that $v(A_1^{t-1}) \leq v(A_2^{t-1})$. Then,
    \begin{equation*}
        c_{t-1}(S,\pi) = v(A_2^{t-1}) - v(A_1^{t-1}).
    \end{equation*}

    We first prove the following claim.

    \noindent
    \textbf{Claim.} Under Algorithm~\ref{alg:greedy_lowest_first} with $n=2$ and goods ordered so that $v(g_1) \leq \dots \leq v(g_m)$, we have $c_{t-1}(S,\pi) \leq v(g_t)$ for every $t \geq 2$.
    Thus, after allocating $g_t$ to the agent with the smaller current bundle value, that agent becomes (weakly) the one with a larger bundle value, and therefore $c_t(S,\pi) = v(g_t) - c_{t-1}(S,\pi)$ (equivalently, $c_t(S,\pi) + c_{t-1}(S,\pi) = v(g_t)$.

    \begin{proof}[Proof of Claim]
        We prove this claim by induction.
        For the base case: $c_1(S,\pi) = v(g_1) \leq v(g_2)$, so $c_1(S, \pi) \leq v(g_2)$.
        For the inductive step: assume $c_{t-1}(S,\pi) \leq v(g_t)$.
        Then, after allocating $g_t$ to the agent with the smaller bundle, $c_t(S,\pi) = v(g_t) - c_{t-1}(S,\pi)$. In particular, $c_t(S,\pi) \leq v(g_t) \leq v(g_{t+1})$.
        Thus, $c_t(S,\pi) \leq v(g_{t+1})$, i.e., $c_t(S,\pi) \leq v(g_{t+1})$.
    \end{proof}
    Thus, Algorithm~\ref{alg:greedy_lowest_first} will allocate $g_t$ to agent~$1$, giving us
    \begin{equation*}
        A_1^t = A_1^{t-1} \cup \{g_t\} \quad \text{and} \quad A_2^t = A_2^{t-1}.
    \end{equation*}
    
    Then, we get that
    \begin{align*}
        c_t(S,\pi) = v(A_1^t) - v(A_2^t) & = v(A_1^{t-1}) + v(g_t)  - v(A_2^{t-1}) \\
        & = - c_{t-1}(S,\pi) + v(g_t).
    \end{align*}
    Thus, for all $t \in \{2,\dots, m\}$,
    \begin{equation*}
        c_t(S,\pi) + c_{t-1}(S,\pi) = v(g_t).
    \end{equation*}
    Together with the fact that $c_1(S,\pi) = v(g_1)$, we get
    \begin{align*}
        v(G) & = v(g_1) + \sum_{t=2}^m v(g_t) \\
        & = c_1(S,\pi) + \sum_{t=2}^m \left( c_t(S,\pi) + c_{t-1}(S,\pi) \right) \\
        & = c_m(S,\pi) + 2 \cdot \sum_{t=1}^{m-1} c_t(S,\pi).
    \end{align*}
    Algebraic manipulation gives us
    \begin{equation*}
        \frac{v(G) + c_m(S,\pi)}{2} = \sum_{t=1}^{m} c_t(S,\pi) = \cost(S,\pi).
    \end{equation*}
    From Lemma~\ref{lem:opt-lb-n}, we know that
    \begin{equation*}
        \OPT(\I) \geq \frac{1}{2} \cdot v(G).
    \end{equation*}
    If $c_m(S,\pi) \leq \frac{v(G)}{2}$, then 
    \begin{align*}
        \frac{\cost(S,\pi)}{\OPT(\I)} & \leq \frac{\left(v(G) + c_m(S,\pi)\right)/2}{v(G) / 2} \\
        & \leq \frac{v(G) + v(G)/2}{v(G)} = \frac{3}{2}.
    \end{align*}
    If $c_m(S,\pi) > \frac{v(G)}{2}$, then given that $\OPT(\I) \geq v(g_m) \geq c_m(S,\pi)$, we have 
    \begin{equation*}
        \OPT(\I) \geq v(g_m) \geq c_m(S,\pi).
    \end{equation*}
    To see why $\OPT(\I) \geq v(g_m)$: For any order $\pi$ and sequence $S$, let $g_\tau$ be a maximum-valued item (so $v(g_\tau)=V_\text{max}=v(g_m))$. If $\tau=1$ then $\cost(S,\pi)\ge c_1=V_\text{max}$. If $\tau\ge 2$, then by the inequality $c_{\tau-1}+c_\tau\ge v(g_\tau)$,\footnote{To see this inequality, let $D_{t-1}=v(A^{t-1}_1)-v(A^{t-1}_2)$. After item $g_t$ is allocated, $D_t=D_{t-1}\pm v(g_t)$. Thus, $c_{t-1}(S,\pi)+c_t(S,\pi)
    =
    |D_{t-1}|+|D_t|
    \ge
    |D_t-D_{t-1}|
    =
    v(g_t)$.} we get $\cost(S,\pi)\ge c_{\tau-1}+c_\tau\ge V_\text{max}$. Thus, $\OPT(\I)\ge V_\text{max}=v(g_m)$.
    
    Using this, we obtain
    \begin{align*}
        \frac{\cost(S,\pi)}{\OPT(\I)} & = 
        \frac{(v(G) + c_m(S,\pi))/2}{\OPT(\I)} \\
        & \leq \frac{(v(G) + c_m(S,\pi))/2}{c_m(S,\pi)} \\
        & = \frac{1}{2} \left( \frac{v(G)}{c_m(S,\pi)} + 1 \right) \\
        & < \frac{1}{2}(2+1) = \frac{3}{2}, 
    \end{align*}
    where the strict inequality follows from the fact that $c_m(S,\pi) > v(G)/2 \implies v(G)/c_m(S,\pi) < 2$.
    
    Thus, in both cases, we get that
    \begin{equation*}
        \cost(S,\pi) \leq \frac{3}{2} \cdot \OPT(\I),
    \end{equation*}
    as desired.

\subsection{Proof of Theorem~\ref{thm:alg_2approx_nagents}}
    Consider any instance $\I=(N,G,\mathbf{v})$ in which agents have identical valuations. 
    Let $\pi=(g_1,\dots,g_m)$ be the order produced by Algorithm~\ref{alg:greedy_lowest_first} (i.e., the goods are labeled so that $v(g_1)\le \cdots \le v(g_m)$), and let $S=(s_1,\dots,s_m)$ be the allocation sequence induced by Algorithm~\ref{alg:greedy_lowest_first} on this order. 
    For each $t\in[m]$ and
    $i\in N$, let $A_i^t$ denote agent $i$'s bundle after the first $t$ goods of $\pi$ have been allocated.
    
    For each $t\in[m]$, define
    \begin{equation*}
        U_t := \max_{i\in N} v(A_i^t) \text{ and } L_t := \min_{i\in N} v(A_i^t).
    \end{equation*}
    Since valuations are identical, for every $t\in[m]$ we have
    \begin{equation*}
        c_t(S,\pi) = \max_{i,j\in N}\max\{0,\,v(A_j^t)-v(A_i^t)\} = U_t - L_t.
    \end{equation*}
    Indeed, for any $i,j\in N$, $v(A_j^t)-v(A_i^t)\le U_t-L_t$, and equality is achieved by choosing $j\in\arg\max_{\ell\in N} v(A_\ell^t)$ and
    $i\in\arg\min_{\ell\in N} v(A_\ell^t)$.
    
    We now prove that
    \begin{equation*}
        c_t(S,\pi)\le v(g_t)\text{ for all }t\in[m],
    \end{equation*}
    by induction on $t$ (this is the bound of envy-cycle elimination~\cite{lipton2004approximately} applied to the first $t$ goods; we include a proof for completeness).
    
    \medskip\noindent
    \emph{Base case ($t=1$).}
    After allocating $g_1$, some agent has value $v(g_1)$, and since $n \geq 2$, at least one agent has value $0$.
    Thus, $U_1 = v(g_1)$ and $L_1 = 0$. 
    Therefore,
    \begin{equation*}
        c_1(S,\pi)=U_1-L_1=v(g_1).
    \end{equation*}
    
    \medskip\noindent
    \emph{Inductive step.}
    Fix $t\in\{2,\dots,m\}$ and assume $c_{t-1}(S,\pi)\le v(g_{t-1})$.
    Since Algorithm~2 sorts goods by nondecreasing value, $v(g_{t-1})\le v(g_t)$, and thus
    \begin{equation*}
        U_{t-1}-L_{t-1}=c_{t-1}(S,\pi)\le v(g_t),
    \end{equation*}
    which implies
    \begin{equation*}
        U_{t-1}\le L_{t-1}+v(g_t).
    \end{equation*}
    Algorithm~\ref{alg:greedy_lowest_first} allocates $g_t$ to some agent $j\in\arg\min_{i\in N} v(A_i^{t-1})$, so $v(A_j^{t-1})=L_{t-1}$. After allocating
    $g_t$ we have
    \begin{equation*}
        v(A_j^t)=v(A_j^{t-1})+v(g_t)=L_{t-1}+v(g_t),
    \end{equation*}
    and for every $i\neq j$, $v(A_i^t)=v(A_i^{t-1})\ge L_{t-1}$. Thus, 
    \begin{equation*}
        L_t=\min_{i\in N} v(A_i^t)\ge L_{t-1}.
    \end{equation*}
    Moreover, since $v(A_j^t)=L_{t-1}+v(g_t)\ge U_{t-1}$ and all other bundle values remain unchanged (and are therefore at most $U_{t-1}$), we have
    \begin{equation*}
        U_t=\max_{i\in N} v(A_i^t)=v(A_j^t).
    \end{equation*}
    Therefore,
    \begin{align*}
        c_t(S,\pi)=U_t-L_t = v(A_j^t)-L_t & = L_{t-1}+v(g_t)-L_t \\
        & \le v(g_t),
    \end{align*}
    where the last inequality follows from the fact that $L_t\ge L_{t-1}$. This completes the induction.
    
    Summing the per-round bound gives us
    \begin{equation*}
        \cost(S,\pi) =\sum_{t=1}^m c_t(S,\pi) \le \sum_{t=1}^m v(g_t) = v(G),
    \end{equation*}
    since $\pi$ is a permutation of $G$. By Lemma~\ref{lem:opt-lb-n}, $\OPT(\I)\ge (1-1/n)\,v(G)$,
    and therefore
    \begin{equation*}
        \frac{\text{cost}(S,\pi)}{\OPT(\I)} \le \frac{v(G)}{(1-1/n)\,v(G)} = \frac{n}{n-1}.
    \end{equation*}
    so Algorithm~\ref{alg:greedy_lowest_first} is a $n/(n-1)$-approximation.

\subsection{Tightness of \textsc{Identical-Lowest-First} Analysis for $n = 2$.} \label{app:eg_tightness_n=2}
We provide an example for which the \textsc{Identical-Lowest-First} algorithm has an approximation ratio of $3/2$ for $n=2$
\begin{example}
    Consider an instance with $n=2$ and $m=3$ where agents have identical valuations and the multiset of good values is $\{1,1,2\}$.
    
    \begin{itemize}
    \item \emph{\textsc{Identical-Lowest-First} (Algorithm~\ref{alg:greedy_lowest_first}):}
    Sorting gives the order $\pi=(1,1,2)$.
    Using the algorithm with its stated tie-breaking rule (lowest-index agent),
    the induced allocation sequence is $S=(1,2,1)$, which gives us
    $c_1(S,\pi)=1, c_2(S,\pi)=0, c_3(S,\pi)=2$, and $\cost(S,\pi)=3$.
    
    \item \emph{Optimal:}
    Let the order be $\pi'=(1,2,1)$ and allocate according to $S'=(1,2,1)$.
    Then, $c_1(S',\pi')=1, c_2(S',\pi')=1, c_3(S',\pi')=0$, and so $\cost(S',\pi')=2$.
    By Lemma~\ref{lem:opt-lb-n} (with $n=2$), $\OPT(\I) \ge (1-\tfrac12)\,v(G)=2$,
    so this is optimal.
    \end{itemize}
    Therefore, on this instance Algorithm~\ref{alg:greedy_lowest_first} has approximation ratio $3/2$, matching the bound in Theorem~\ref{thm:alg_3/2approx}.
\end{example}
\subsection{Approximation Ratio Lower Bound of $\frac{2n-1}{2n-2}$ for $n \geq 3$.} \label{app:eg_tightness_ngeq3}
We provide an example for which the \textsc{Identical-Lowest-First} algorithm has an approximation ratio of $\frac{2n-1}{2n-2}$ for $n \geq 3$.

\begin{example}\label{ex:lb-nge3}
Fix any $n\ge 3$. Consider an instance with identical valuation function $v$ and $m=n+1$ goods. 
Let two goods have value $1$ and $n-1$ goods have value $2$.
Thus, we have that $v(G)=2n$.

\begin{itemize}
\item \emph{\textsc{Identical-Lowest-First} (Algorithm~\ref{alg:greedy_lowest_first}):} 
Sorting gives the order \\ $\pi=(1,1,2,2,\dots,2)$.
Using the algorithm with its stated tie-breaking rule (lowest-index agent), after allocating the first two goods of value $1$, we have bundle values $(1,1,0,\dots,0)$.
Then the next $n-2$ goods each of value $2$ are allocated one by one to agents $3,4,\dots,n$, giving us bundle values $(1,1,2,\dots,2)$ after round $n$; the final $2$ is then allocated to agent $1$, giving us $(3,1,2,\dots,2)$ after round $n+1$.
Let $S$ be the induced allocation sequence.

Under identical valuations, the per-round cost is the gap between the maximum and minimum bundle values. 
Thus the maximum envy values are $c_1(S,\pi)=1$, $c_2(S,\pi)=1$, $c_t(S,\pi)=2$ (for $t = 3,\dots,n-1$, $c_n=1$), and $c_{n+1}(S,\pi)=2$. 
Therefore,
\begin{align*}
    \cost(S,\pi)
    & =\sum_{t=1}^{n+1} c_t(S,\pi) \\
    & =1+1+2(n-3)+1+2 = 2n-1.
\end{align*}

\item \emph{Optimal solution.}
Consider the arrival order $\pi^*=(1,2,2,\dots,2,1)$,
and allocate the first and last value $1$ goods to agent $1$, while allocating the $n-1$ value $2$ goods to agents $2,3,\dots,n$ (one each). Denote the resulting allocation sequence by $S^*$. Then the per-round maximum
envy values are $c_1(S^*,\pi^*)=1$, $c_t(S^*,\pi^*)=2$ for $t=2,\dots,n-1$, $c_n(S^*,\pi^*)=1$, $c_{n+1}(S^*,\pi^*)=0$.
Consequently, $\cost(S^*,\pi^*)=1+2(n-2)+1+0=2n-2$.
By Lemma~\ref{lem:opt-lb-n}, $\OPT(\I)\ge (1-\frac1n)\,v(G)=(1-\frac1n)\cdot 2n=2n-2$, giving us
$\OPT(\I)=2n-2$.
\end{itemize}

Thus, on this instance, Algorithm~\ref{alg:greedy_lowest_first} has approximation ratio
\[
\frac{\cost(S,\pi)}{\OPT(\I)}
=\frac{2n-1}{2n-2}.
\]
\end{example}

\subsection{Proof of Theorem~\ref{thm:alg_1/2approx_additive}}
Fix an instance $\I=(N,G,v)$ with identical valuations, where $n\ge 2$.
Let Algorithm~\ref{alg:greedy_lowest_first} order the goods as $\pi=(g_1,\dots,g_m)$ with
$v(g_1)\le \cdots \le v(g_m)$, and let $S$ be the induced allocation sequence.
For each $t\in\{0,1,\dots,m\}$ define
\[
L_t := \min_{i\in N} v(A_i^t) \quad \text{and} \quad
U_t := \max_{i\in N} v(A_i^t)
\]
where $A_i^t$ denotes agent $i$'s bundle after round $t$ (and $L_0=U_0=0$).
Under identical valuations, $c_t(S,\pi) = U_t - L_t$.

We claim that $c_t(S,\pi)\le v(g_t)$ for all $t\in[m]$.
For $t=1$, after allocating $g_1$ we have $U_1=v(g_1)$ and $L_1=0$ (since $n\geq 2$ means some agent is still empty after round $1$), thus, $c_1(S,\pi)=v(g_1)$.
Now fix $t\ge 2$ and assume inductively that $c_{t-1}(S,\pi)\le v(g_{t-1})$.
Since the goods are sorted, $v(g_{t-1})\le v(g_t)$, so $c_{t-1}(S,\pi) \le v(g_t)$.
Algorithm~\ref{alg:greedy_lowest_first} allocates $g_t$ to an agent $j$ with minimum bundle value at time $t-1$, i.e.,
$v(A_j^{t-1})=L_{t-1}$.
After allocation,
\[
v(A_j^{t}) = L_{t-1}+v(g_t) \ge\ L_{t-1}+c_{t-1}(S,\pi) =\ U_{t-1},
\]
so $j$ becomes a (weak) maximizer and therefore $U_t = L_{t-1}+v(g_t)$.
Also $L_t\ge L_{t-1}$.
Thus
\[
c_t(S,\pi) = U_t - L_t = (L_{t-1}+v(g_t)) - L_t \le v(g_t),
\]
proving the claim.

In particular, $c_m(S,\pi) \le v(g_m)=V_\text{max}$.

From $U_t=L_{t-1}+v(g_t)$ we have, for every $t\in[m]$, $c_t(S,\pi) = U_t-L_t = (L_{t-1}+v(g_t)) - L_t$; equivalently, $v(g_t)=c_t(S,\pi)+(L_t-L_{t-1})$.

Summing over $t=1,\dots,m$ gives us
\begin{align*}
    v(G) = \sum_{t=1}^m v(g_t) & = \sum_{t=1}^m c_t(S,\pi)+ \sum_{t=1}^m (L_t-L_{t-1}) \\
    & = \cost(S,\pi) + L_m,
\end{align*}
and thus
\begin{equation}\label{eq:cost-vG-Lm}
    \cost(S,\pi)=v(G)-L_m.
\end{equation}

Let $i^* \in \argmin_{i \in N} v(A_i^m)$, so $v(A_{i^*}^m) = L_m$, while every other bundle is $\leq L_m + c_m$.
Thus,
\begin{align*}
    v(G)=\sum_{i\in N} v(A_i^m) & \le L_m + (n-1)(L_m+c_m) \\
    & = nL_m + (n-1)c_m.
\end{align*}
Rearranging, we get
\[
L_m \ge \frac{v(G) - (n-1)c_m}{n}.
\]
Plugging this into \eqref{eq:cost-vG-Lm} gives us
\begin{align*}
    \cost(S,\pi) = v(G)-L_m & \le v(G) - \frac{v(G)-(n-1)c_m}{n} \\
    & = \left(1-\frac{1}{n}\right)v(G) + \left(1-\frac{1}{n}\right)c_m.
\end{align*}
Using $c_m\le V_\text{max}$, we get
\begin{equation*}
    \cost(S,\pi) \le \left(1-\frac{1}{n}\right)v(G) + \left(1-\frac{1}{n}\right)V_\text{max}.
\end{equation*}

By Lemma~\ref{lem:opt-lb-n}, $\OPT(\I)\ge \left(1-\frac{1}{n}\right)v(G)$.
Therefore,
\[
    \cost(S,\pi) \le \OPT(\I) + \left(1-\frac{1}{n}\right)V_\text{max},
\]
which completes the proof.

\subsection{Tightness of the Additive Approximation Guarantee.} \label{app:eg_tightness_additive}
We provide an example showing that the additive error term in Theorem~\ref{thm:alg_1/2approx_additive} is tight.

\begin{example}
    Fix any $n\ge 2$. Consider an instance with identical valuation function $v$ and the following multiset of good values:
    for each $k\in\{1,\dots,n-1\}$ there are exactly $n$ goods of value $k$, and there are $n-1$ goods of value $n$.
    Thus $V_\text{max}=n$, the number of goods is $m=n(n-1)+(n-1)=n^2-1$, and
    \begin{align*}
        v(G)=n\sum_{k=1}^{n-1}k + (n-1)\cdot n
    & = n\cdot \frac{n(n-1)}2 + n(n-1)\\
    & = \frac{n(n-1)(n+2)}2.
    \end{align*}
    
    \begin{itemize}
    \item \emph{\textsc{Identical-Lowest-First} (Algorithm~\ref{alg:greedy_lowest_first})}
    Sorting gives the order
    \[
    \pi=(\underbrace{1,\dots,1}_{n},\underbrace{2,\dots,2}_{n},\dots,\underbrace{n-1,\dots,n-1}_{n},\underbrace{n,\dots,n}_{n-1}).
    \]
    For each $k\in\{1,\dots,n-1\}$, when the $n$ goods of value $k$ arrive, all bundle values are equal at the start of this block,
    so the greedy rule gives one such good to each agent. Thus, during this block, the gap (maximum minus minimum bundle value) is $c_t(S,\pi) =k$ for the first $n-1$ arrivals and $c_t(S,\pi)=0$ on the $n$-th, contributing $(n-1)k$ to the cumulative cost.
    After the first $n(n-1)$ arrivals, every agent has value
    \[
    \sum_{k=1}^{n-1}k = \frac{n(n-1)}2.
    \]
    The remaining $n-1$ goods (each of value $n$) are then given to $n-1$ distinct agents; after each such round, the gap is exactly $n$.
    Therefore,
    \begin{align*}
        \cost(S,\pi)
        & =\sum_{k=1}^{n-1}(n-1)k + (n-1)\cdot n \\
        & =(n-1)\left(\sum_{k=1}^{n-1}k+n\right) \\
        & =(n-1)\left(\frac{n(n-1)}2+n\right) \\
        & =\frac{n(n-1)(n+1)}2.
    \end{align*}
    \item \emph{Optimal solution.}
    We show a pair $(S^*,\pi^*)$ that achieves the universal lower bound of Lemma~\ref{lem:opt-lb-n}.
    Define an arrival order $\pi^*$ in $(n-1)$ blocks as follows: for each $k=n-1,n-2,\dots,1$, append the block
    \[
    \pi^* = (k,\;\underbrace{(k+1),(k+1),\dots,(k+1)}_{k\text{ times}},\;
    \underbrace{k,k,\dots,k}_{(n-k)\text{ times}}).
    \]
    Note that for each $k \leq n-1$, value $k$ appears $n-k+1$ times in block $k$ (once at the start, and $n-k$ times at the end), and appears $k-1$ times in block $k-1$. Thus, the total number of times value $k$ appears is $(n-k+1) + (k-1) = n$.
    Also, value $n$ appears exactly $n-1$ times in block $n-1$.
    
    Let $S^*$ allocate each arriving good to an agent with currently minimum bundle value (break ties arbitrarily).
    
    \smallskip
    \noindent\textbf{Claim.} After the block for $k$ is processed, the multiset of bundle values is
    \[
    \{\underbrace{L_k,\dots,L_k}_{k\text{ agents}},\underbrace{L_k+(k-1),\dots,L_k+(k-1)}_{n-k\text{ agents}}\},
    \]
    where $L_{n-1}=n$ and $L_{k-1}=L_k+k$ for $k\in\{2,\dots,n-1\}$.
    
    \smallskip
    \noindent\emph{Proof of claim.}
    For $k=n-1$, the first block is $(n\!-\!1,\underbrace{n,\dots,n}_{n-1},n\!-\!1)$, and allocating each good to a minimum-valued
    agent yields $n-1$ agents at value $n$ and one agent at value $2n-2=n+(n-2)$, so $L_{n-1}=n$.
    For the inductive step, assume the claim holds after the block for $k$ (so the current gap is $k-1$).
    The next block is
    \[
    \big(k-1,\underbrace{k,\dots,k}_{k-1},\underbrace{k-1,\dots,k-1}_{n-k+1}\big).
    \]
    The first $(k-1)$ goes to a minimum agent, bringing her up to the current maximum.
    Then the $(k-1)$ goods of value $k$ are given to the remaining minimum agents, making them the new maxima.
    Finally, each of the $n-k+1$ goods of value $(k-1)$ is given to a (new) minimum agent; after all are assigned,
    we obtain $k-1$ agents at value $L_k+k$ and $n-k+1$ agents at value $(L_k+k)+(k-2)$, i.e., the claim with $L_{k-1}=L_k+k$.
    \hfill$\square$
    
    \smallskip
    In particular, after the final block ($k=1$), all agents have equal value $L_1$. Unrolling the recurrence,
    \[
    L_1 = n+(n-1)+\cdots+2 = \frac{(n-1)(n+2)}2 = \frac{v(G)}n.
    \]
    Moreover, each arriving good is allocated to an agent with minimum bundle value.
    As in the proof of Lemma~\ref{lem:opt-lb-n}, this implies
    \begin{align*}
        \cost(S^*,\pi^*) & = v(G)-\min_i v(A_i) \\
        & = v(G)-\frac{v(G)}n=\left(1-\frac1n\right)v(G).
    \end{align*}
    By Lemma~\ref{lem:opt-lb-n}, $\OPT(\I)\ge (1-\frac1n)v(G)$, so we conclude $\OPT(\I)=(1-\frac1n)v(G)$.
    \end{itemize}
    Finally, we have that
    \begin{align*}
        & \cost(S,\pi)-\OPT(\I) \\
        & =\left(v(G)-\frac{n(n-1)}2\right)-\left(v(G)-\frac{v(G)}n\right) \\
        & =\frac{v(G)}n-\frac{n(n-1)}2 \\
        & =\frac{(n-1)(n+2)}2-\frac{n(n-1)}2 \\
        & =n-1 =\left(1-\frac1n\right)V_\text{max}.
    \end{align*}
    Thus, the additive term $\left(1-\frac1n\right)V_\text{max}$ cannot be improved in general.
\end{example}

\end{document}